\newif\iffull
\fulltrue    

\iffull
\documentclass[12pt]{article}
\usepackage[margin=1in]{geometry}
\else
\documentclass[11pt]{article}
\usepackage[margin=.6in,bottom=0.2in,foot=0.3in,includefoot]{geometry}
\fi

\usepackage{style/qstyle}

\iffull\else
\usepackage{titling}
\fi

\begin{document} 

\iffull
\title{\textbf{Uncloneable Cryptographic Primitives \\ with Interaction}}
\fi
\author[1]{Anne Broadbent}
\author[2]{Eric Culf}
\date{\vspace{-10mm}} 
\affil[1]{University of Ottawa,  Department of Mathematics and Statistics \footnote{\texttt{abroadbe@uottawa.ca}}}
\affil[2]{University of Waterloo, Institute for Quantum Computing and Faculty of Mathematics \footnote{\texttt{eculf@uwaterloo.ca}}
}
\setcounter{Maxaffil}{0}
\renewcommand\Affilfont{\small}

\maketitle
\iffull
\begin{abstract}
	Much of the strength of quantum cryptography may be attributed to the no-cloning property of quantum information. We construct three new cryptographic primitives whose security is based on uncloneability, and that have in common that their security can be established via a novel monogamy-of-entanglement (MoE) property:
	\begin{itemize}
		\item We define \emph{interactive uncloneable encryption}, a version of the uncloneable encryption defined by Broadbent and Lord [TQC 2020] where the receiver must partake in an interaction with the sender in order to decrypt the ciphertext. We provide a one-round construction that is secure in the information-theoretic setting, in the sense that no other receiver may learn the message even if she eavesdrops on all the interactions.
		
		\item We provide a way to make a bit string commitment scheme uncloneable. The scheme is augmented with a check step chronologically in between the commit and open steps, where an honest sender verifies that the commitment may not be opened by an eavesdropper, even if the receiver is malicious. Our construction preserves the assumptions of the original commitment  while requiring only a polynomial decrease in the length of the committed string.
		\item We construct a \emph{receiver-independent quantum key distribution} (QKD) scheme, which strengthens the notion of one-sided device independent QKD of Tomamichel, Fehr, Kaniewski, and Wehner (TFKW) [NJP 2013] by also permitting the receiver's classical device to be untrusted. Explicitly, the sender remains fully trusted while only the receiver's communication is trusted. We provide a construction that achieves the same asymptotic error tolerance as the scheme of TFKW.
	\end{itemize}
	To show security, we prove an extension of the MoE property of coset states introduced by Coladangelo, Liu, Liu, and Zhandry [Crypto 2021]. In our stronger version, the player Charlie also receives Bob's answer prior to making his guess, thus simulating a party who eavesdrops on an interaction. To make use of this property, we express it as a new type of entropic uncertainty relation which arises naturally from the structure of the underlying MoE~game.
	
\end{abstract} 
\fi

\iffull
\newpage
\setcounter{tocdepth}{2}
\tableofcontents
\newpage
\fi


\iffull
\section{Introduction}
\fi

An important feature of quantum information is the no-cloning principle --- the property that an arbitrary quantum state cannot be perfectly copied, unlike a classical string~\cite{Par70,WZ82,Die82}. This idea underpins many of the unique constructions in quantum cryptography~\cite{BS16}, beginning with quantum money~\cite{Wie83} and quantum key distribution (QKD)~\cite{BB84}. In this work, we give three new constructions of cryptographic primitives that, at the intuitive level, make use of uncloneability: \emph{uncloneable encryption with interactive decryption}, \emph{uncloneable bit commitment}, and \emph{receiver-independent QKD}. An important consequence of the uncloneability is that none of these primitives can be secure classically --- in fact, as classical information can always be copied, the security is clearly unachievable.

In order to prove security of these primitives and formally  reason about their ``uncloneability,'' we show a strengthened form of the subspace coset state monogamy-of-entanglement (MoE) property \cite{CLLZ21, CV22}, which is a bound on the winning probability of an MoE game built using subspace coset states. MoE games are used to quantify the strength of quantum tripartite correlations. They belong to the family of \emph{extended nonlocal games} \cite{JMRW16}, which generalise nonlocal games, but are highly distinct from them. The MoE game paradigm, introduced in \cite{TFKW13}, has recently been used in various uncloneability-related cryptographic constructions \cite{BL20,BC23,CLLZ21}. An MoE game is played between two cooperating players, Bob and Charlie, and an honest referee, Alice, all of whom may hold a quantum system. The subspace coset MoE game (called the strong monogamy game in \cite{CLLZ21}), proceeds as follows. First, Alice samples a subspace $a$ of dimension $n/2$ of the space of $n$-bit strings $\Z_2^n$, and strings $t,t'$ uniformly at random, and prepares the coset state\footnote{We use lowercase rather than uppercase letters for subspaces as we aim to reserve the uppercase letters for registers and random variables.}
\iffull
\begin{align}
	\ket{a_{t,t'}}=\frac{1}{\sqrt{|a|}}\sum_{u\in a}(-1)^{u\cdot t'}\ket{t+u}.
\end{align}
\fi
 She sends this state to Bob and Charlie, who may do arbitrary actions to split\footnote{Note that the splitting operation is represented by an arbitrary quantum channel, chosen by Bob and Charlie. It is not necessarily something simple like a bipartition of the qubits.} the state between their two systems, after which they are isolated. Next, Alice provides them with a description of~$a$. In order to win, Bob must provide a vector from the coset $t+a$ and Charlie must provide one from $t'+a^\perp$, where $a^\perp$ is the orthogonal complement of $a$. This game was shown in~\cite{CV22} to have an exponentially small winning probability in $n$. We strengthen the relation by showing that the same bound holds on a version of the game that is easier to win --- Bob's answer, whether or not it is correct, leaks to Charlie before he makes his guess. In this way, we are able to see the information that Charlie gets as messages sent during an interaction between Alice and Bob, on which he eavesdrops. We refer to this bound on the winning probability as the \emph{\gamename}monogamy-of-entanglement property.
\iffull

\else
\fi

\iffull
\subsection{Uncloneable encryption with interactive decryption}\label{sec:intro-encryption}
\fi

We introduce, study, and construct a variant of uncloneable encryption that allows for an interaction during the decryption process. Uncloneable encryption as is currently understood was introduced in \cite{BL20}, building on earlier concepts such as the tamper-evident encryption of \cite{Got03} and the MoE games of \cite{TFKW13}. In its most general form, an uncloneable encryption scheme provides a way to encrypt messages in such a way that they cannot be simultaneously read by two malicious parties, Bob and Charlie, under the assumption that they are isolated once the encryption key is released. To the best of our knowledge, it is unknown whether this is achievable in the plain model, even if we allow computational assumptions. Uncloneable encryption schemes in the quantum random oracle model (QROM) have been studied \cite{BL20} and provide nearly optimal security. Other computational assumptions have been considered: under the assumption of post-quantum one-way functions, \cite{AK21} show that it is possible to turn an uncloneable encryption scheme into one with semantic security; and under the assumption of a post-quantum public key encryption scheme, they show how to turn the scheme into a public-key uncloneable encryption scheme. Since all these rely on the existence of uncloneable encryption, a key open question remains concerning the existence of an ``uncloneable bit'' --- an optimal uncloneable encryption scheme in the plain model that encrypts one-bit message. This is a fundamental object as any uncloneable encryption scheme implies an uncloneable bit \cite[Theorem 9]{BL20}. We work with a simple communication assumption rather than a computational assumption in order to instantiate a new form of uncloneable encryption.\looseness=-1

Originally, the encryption was represented by a quantum encryption of classical messages (QECM), a protocol that encrypts classical messages as quantum ciphertexts, which can be decrypted using only the classical encryption key \cite{BL20}. A QECM scheme is uncloneable if two receivers receive a ciphertext, split it arbitrarily, and only get the key once they are isolated, then they can simultaneously learn the message with at best near-trivial probability. We extend the original non-interactive setting of \cite{BL20} by allowing interaction in the decryption phase. We call this model \emph{quantum encryption of classical messages with interactive decryption} (QECM-ID). To adapt uncloneability to a QECM-ID scheme, we again have two receivers, whom we call Bob and Eve, who split a ciphertext. To decrypt, Bob initiates an interaction with Alice. Only after this point does Bob need to be seen as the intended recipient of the message. To avoid the trivial attack where Bob simply gives the decrypted message to Eve, they may not communicate directly during the interaction step --- nevertheless, Eve may eavesdrop on the communication between Alice and Bob. We therefore say that the encryption is uncloneable if, for any actions Bob and Eve take, the probability that Eve guesses the message correctly once the interaction finishes and the decryption protocol does not abort is near-trivial.

 We also adapt uncloneable-indistinguishable security, which is meant to represent an uncloneability version of chosen-plaintext attack (CPA) security. For a QECM, this is the property that Bob and Eve cannot simultaneously distinguish the encryption of a chosen message distribution from a fixed message~\cite{BL20}. To adapt this to a QECM-ID, we say that it is uncloneable-indistinguishable secure if, after the decryption interaction, the probability that, simultaneously, Alice accepts the decryption and Eve distinguishes a chosen message distribution from a fixed message is near trivial, \emph{i.e.}~half the probability of accepting. Intuitively, the condition that Bob guesses correctly is replaced with the condition that Alice accepts the decryption in order to adapt the definition to a QECM-ID.

 Finally, we show that there is an equivalence between uncloneable and uncloneable-indistinguishable security for QECM-IDs. This extends the result, shown in~\cite{BL20}, that uncloneable security implies uncloneable-indistinguishable security for QECMs. Further, the equivalence generalises an important property of classical encryption. To the best of our knowledge, it is unknown whether both implications hold for QECMs.

\iffull
\paragraph{Proof technique.} To instantiate an uncloneable QECM-ID, we make use of the \gamename MoE property. Alice, to encrypt her message $m$, uses as a key a subspace $a$, strings $h$ and $t,t'$, and a key $r$ for a quantum-proof strong extractor $e$ . She sends the pair $(m+e(t',r)+h,\ket{a_{t,t'}})$ as the ciphertext. The MoE property implies that, if Bob is able to provide $t$ to Alice, then with high probability Eve is unable to guess $t'$ correctly, even if she learns $t$. Hence, Alice can use the interaction to check whether Bob knows $t$. If this succeeds, then $t'$ is secure against Eve with high probability, so Alice sends remainder of the key $(r,h)$ to Bob. With this, our construction satisfies both forms of uncloneable security, with tighter bounds that the equivalence between the properties implies.

\else

\fi

\iffull
\subsection{Uncloneable bit commitment}\label{sec:intro-commitment}
\fi

In bit string commitment, a sender Alice commits to a string that a receiver Bob can only access when she chooses. Ideally, the commitment should be \emph{hiding}, in the sense that Bob cannot learn the string Alice has committed until she chooses to reveal, and \emph{binding}, in the sense that Alice must reveal the same string to which she had committed. Without additional assumptions, bit commitment is impossible \cite{May96,LC97,BS16}, but there are a variety of models in which it was shown to exist. For example, under classical computational assumptions \cite{Cha87,Nao91} (see also \cite{Cre'11}) or in the noisy quantum storage model \cite{KWW12}. However, a problem underlying many classically-defined cryptographic primitives is that they are inherently cloneable; if an eavesdropper Eve is able to eavesdrop on the communications between Alice and Bob, she may be able to produce a transcript of their interactions and hence learn the final string whenever it is revealed. This is the case for bit commitment: in fact, the reveal step is usually represented as a public broadcast with no indication of security against an eavesdropper. We remedy this with a method to make a bit string commitment scheme uncloneable.

We define an \emph{uncloneable bit string commitment scheme} as a commitment scheme with an additional check step in between the commit and reveal steps, where Alice verifies whether an eavesdropper has attempted to clone the commitment. If the commitment passes this check, then an honest Alice can be sure that only Bob will be able to open it during the reveal phase, despite a lack of prior agreement between them. Bob may even be malicious: the only restriction needed on him is that he does not communicate directly to Eve after the check. With this in mind, the point in time when Alice chooses to undertake the check allows it to be run under varying assumptions. In particular, Alice may check immediately after committing, which means that no honest party needs to store any quantum information, but Alice needs to be sure that Bob does not communicate privately with Eve at any point after committing. This is more feasible for near-term quantum devices, but requires that Bob not communicate information to Eve for a period of time between steps. On the other hand, if Alice waits until immediately before revealing to do the check, she may assume that Bob and Eve have arbitrary communication after committing. The drawback is that Bob must store a quantum state even if he is honest.

\iffull
\paragraph{Proof technique.} We use the \gamename MoE property to provide a way to turn a commitment scheme into an uncloneable commitment of the above form, which works under the same assumptions as the original commitment. We assume that this is a randomised commitment scheme, where Alice commits to a uniformly random string; this form of commitment is equivalent to the standard one where Alice chooses the string to commit \cite{KWW12}. In order to commit to the random string $e(t',r)+h$, where $e$ is a quantum-proof strong extractor, Alice commits to $(r,h)$ using the original commitment and sends a coset state $\ket{a_{t,t'}}$ to Bob. Because Bob does not know $a$, he has no information about $t'$ and $(r,h)$ has not been revealed, so the commitment is hiding. Next, to check for cloning, Alice sends $a$ to Bob and verifies that he can measure $t$. Due to the \gamename MoE property, this implies that Eve is only able to guess $t'$ with low probability. Finally, to reveal, Alice reveals $(r,h)$ and Bob queries Alice for some information about $t'$ to make sure that their values are consistent, making the scheme binding. With a good choice of strong extractor, this causes only a polynomial decrease in the length of the committed string and an exponentially small change in the binding parameter.
\else

\fi

\iffull
\subsection{Receiver-independent QKD}\label{sec:intro-qkd}
\fi

Quantum key distribution (QKD), introduced by Bennett and Brassard \cite{BB84}, is a foundationally important quantum cryptographic primitive. In its most basic form, it allows an honest sender, Alice, to share a secret key with an honest receiver, Bob, over a public channel without an eavesdropper Eve learning the key. Many variants of QKD that require only weaker assumptions on the honest parties have been proposed. In particular, device-independent protocols, initiated by Ekert~\cite{Eke91}, seek to allow QKD with few, if any, assumptions on the behaviour of Alice and Bob's devices. One-sided device-independent QKD, shown to be secure against any eavesdropper in \cite{TFKW13}, allows Bob's quantum device to be fully untrusted, relying on a monogamy-of-entanglement game winning probability bound for security; and fully device-independent QKD, shown by Vazirani and Vidick \cite{VV14b}, allows both Alice and Bob's quantum devices to be untrusted, with security coming from the rigidity of a nonlocal game. These varying assumptions allow implementations of QKD to balance practicality and security, depending on available resources.

We show security of QKD in a model extending the one-sided device-independent model, which we call \emph{receiver-independent QKD}. In this model, Alice's quantum device remains fully trusted, but neither Bob's quantum nor his \emph{classical} device is trusted. However, we require that Bob's communication be trusted: if Bob's communication were not trusted, any QKD scheme would be susceptible to the trivial attack where Bob sends his final key to Eve. In this way, this model can be seen as the minimal assumption on the receiver, hence warranting the name ``receiver-independent''.
\iffull

\else \fi
Receiver-independent QKD schemes are distinct in a number of ways. First, since any computation Bob might want to make is inherently untrusted, he cannot be trusted to check any property of the shared state. As such, only Alice may be given the power to abort the protocol. In this way, the interactions between Alice and Bob take the form of a sequence of challenges and responses. Also, the idea of correctness must be altered to account for the fact that Bob's classical computations are untrusted. This is because it is not possible to be certain that Bob has access to the final key, but it is possible to be sure that his device can compute it.

\iffull
\paragraph{Proof technique.} We construct a receiver-independent QKD scheme using coset states, and show its security using an error-robust generalisation of the \gamename MoE property. Alice sends a coset state $\ket{a_{t,t'}}$ to Bob. To verify that Eve does not have $t'$, Alice asks Bob to provide $t$, acting as the parameter estimation step. If he is able to, with only small error, then Alice issues challenges to Bob that allow her to correct her $t'$ to match the guess $\hat{t}'$ Bob's device claims to have, and then verify this match, which act as the error correction and information reconciliation steps, respectively. Finally, for privacy amplification, Alice acts on her corrected raw key with a quantum-proof strong extractor and instructs Bob to do the same. It is worth noting that our use of an entropic uncertainty relation, as introduced in \cref{sec:eur-intro} below, brings the security proof intuitively closer to earlier proofs of QKD security, as in~\cite{Ren05}, than the proof of~\cite{TFKW13}, which works more directly with an MoE game.
\else

\fi

\iffull
\subsection{Main technique: MoE entropic uncertainty relations}
\fi
\label{sec:eur-intro}
Entropic uncertainty relations, and earlier uncertainty relations beginning with~\cite{Hei27}, have played a foundational role in quantum information~\cite{WW10}. Tomamichel, Fehr, Kaniewski, and Wehner show an entropic uncertainty relation in the same scenario as their MoE game~\cite{TFKW13}. We provide an entropic uncertainty relation that arises naturally from the scenario of the \gamename subspace coset MoE game, allowing us to work with the full strength of the MoE property in an entropy setting.

To show our relation, we generalise the min-entropy of guessing $H_{\min}(X|A)_\rho$ to a novel property that we refer to as the \emph{sequential min-entropy}, $H_{\min}(X|A;Y|B)_\rho$, which represents the uncertainty of guessing~$X$ knowing $A$, followed by guessing~$Y$ knowing $B$, on the same state. For any measurement $M$ on $A$ used to guess $X$, this decomposes as the entropic uncertainty relation
\iffull
\begin{align}
	H_{\min}(X|M(A))_\rho+H_{\min}(Y|B)_{\rho_{|(M(A)=X)}}\geq H_{\min}(X|A;Y|B)_\rho,
\end{align}
\fi
where $\rho_{|(M(A)=X)}$ is the state conditioned on the guess of $X$ being correct. A notable distinction between such an entropic uncertainty and a more standard relation is that the states on the two terms are different, although closely related. The winning probability of the \gamename MoE game can directly be expressed using a sequential entropy as $\exp\parens*{-H_{\min}(T|AB,T'|A'TC)_\rho}$, where $\rho_{AA'TT'BC}$ is the state such that $A$ and $A'$ hold two copies of the subspace $a$, $T$ and $T'$ hold the coset representatives $t,t'$, and $B$ and $C$ hold Bob and Charlie's quantum systems once they are isolated. Hence, the \gamename MoE property provides the entropic uncertainty relation
\iffull
\begin{align}
	H_{\min}(T|M(AB))_\rho+H_{\min}(T'|A'TC)_{\rho_{|(M(AB)=T)}}\in \Omega(n).
\end{align}
\fi
This may be compared to the MoE~game-based entropic uncertainty relation that was studied in~\cite{TFKW13}, $H_{\min}(X|\Theta B)_{\rho}+H_{\min}(X|\Theta C)_{\rho}\geq-2\lg\squ*{(1+2^{-n/2})/2}\in O(1)$, where $\rho_{ABC}$ is any quantum state with $A=\Z_2^n$, $X$ is the result of measuring $A$ in a uniformly random Wiesner basis of states $\ket{x^\theta}=H^{\theta_1}\ket{x_1}\otimes\cdots H^{\theta_n}\ket{x_n}$, and~$\Theta$ is the description of the basis. The relation is found in the same way as their bound on the winning probability of their MoE game, but is strictly weaker than that bound, since it only considers entropies with respect to the same state. This makes it too weak to provide security of cryptographic primitives such as QKD. In fact, even in the case of the subspace coset MoE game, we similarly have
\iffull
\begin{align}
	H_{\min}(T|M(AB))_\rho+H_{\min}(T'|A'TC)_\rho\in O(1),
\end{align}
\fi
using the same simple attack: half the time, Bob takes the whole state, and the other half of the time, Charlie takes the whole state.

In order to extend the use of the \gamename MoE property and associated entropic uncertainty relation to scenarios where errors should be accounted for, such as QKD, we adapt the MoE game to allow for errors. That is, we show a bound on the winning probability of a robust generalisation of the \gamename MoE game where Bob and Charlie's answers are considered to be correct even if some small number of bits are wrong. The important case for QKD is where Bob is allowed to guess $t$ incorrectly up to relative error $\gamma$ but Charlie, who represents the eavesdropper, must still answer perfectly. For small enough error, the winning probability remains exponentially small in $n$. We can also handle this probability of approximate guessing as an entropic uncertainty relation, by representing the ``entropy of approximate guessing'' as an entropy of exact guessing on a modified state. Explicitly, the relation takes the now-familiar form
\iffull
\begin{align}
	H_{\min}(T|M(AB))_\sigma+H_{\min}(T'|A'TC)_{\sigma_{|(M(AB)=T)}}\in \Omega(n),
\end{align}
\fi
where $\sigma$ is the state modified to account for the error bit flips $\sigma=\expec_{|u|\leq\gamma n/2}X_T^u\rho X_T^u$.

\iffull
\subsection{Further related work}

The no-cloning property is found in a wide and growing range of cryptographic applications, such as tamper-detection \cite{Got03}, copy-protection \cite{Aar09,CMP20arxiv}, certified deletion \cite{BI20}, secure software leasing \cite{AL21,BJL+21}, and uncloneable decryption \cite{CLLZ21}.

The coset states we study act as a generalisation of subspace states --- uniform superpositions of the elements of a subspace --- introduced in the context of quantum money by Aaronson and Christiano \cite{AC12}. Rather than using the properties of subspaces, it is possible to see the generalisation to coset states as subspace states encrypted with a quantum one time pad $\ket{a_{t,t'}}=X^tZ^{t'}\ket{a}$. Coset states under this definition have been studied in the context of proofs of knowledge by Vidick and Zhang \cite{VZ21}.

Though inspired by uncloneable encryption of \cite{BL20}, uncloneable encryption using a QECM-ID also bears comparison to tamper-evident encryption, introduced by Gottesman \cite{Got03} (under the name uncloneable encryption). This is a scheme where an honest receiver can verify, during decryption, whether an eavesdropper had attempted to clone an encrypted message. We emphasize that \cite{Got03} requires both an honest sender and receiver and that our techniques are fundamentally different since they are resilient to a dishonest receiver.

Finally, the recent work of Kundu and Tan \cite{KT22arxiv} provides an alternate extension of the uncloneable encryption paradigm. They consider the case where, for each encryption key, there are multiple decryption keys. They give a construction of an encryption scheme that is uncloneable as long as the attackers receive independently generated keys. Similarly to the interaction in our model, an assumption on the communication during the decryption is used to guarantee uncloneability. Also, their results consider noise on the devices, similarly to what we are concerned with in the robust version of the game used for receiver-independent QKD; arbitrary small leakage of information between Bob and Charlie's devices, contrasting with our fixed but large leakage of Bob's measurement result; and full device-independence, which requires an interactive encryption

\subsection{Acknowledgements}

This work was supported by the Air Force Office of Scientific Research under award number FA9550-20-1-0375, Canada's NSERC, and the University of Ottawa's Research Chairs program.

\subsection{Outline}

In \cref{sec:prelims}, we introduce our notation and the relevant basic technical facts. In \cref{sec:moe}, we introduce and analyse the monogamy-of-entanglement game we study, as well as the related entropic uncertainty relation. In \cref{sec:encryption,,sec:commitment,,sec:qkd} we define and study the primitives of interactive uncloneable encryption, uncloneable bit commitment, and receiver-independent QKD, respectively. In \cref{sec:qkd}, we also study the robust version of the MoE game. The MoE properties are given as \cref{thm:stronger-monogamy} and \cref{thm:robust-stronger-monogamy}, and their expressions as entropic uncertainty relations as \cref{cor:entropic-moe} and \cref{cor:robust-entropic-moe}.
\fi 

\iffull
\section{Preliminaries}\label{sec:prelims}

In this section, we introduce the notation and recall the technical facts we use in this paper. In \cref{sec:states}, we go over the basics of quantum information and probability that we need; in \cref{sec:coset}, we discuss subspaces of vector spaces of bit strings and recall the definition of subspace coset states; and in \cref{sec:entropy}, we note the definitions of conditional min-entropy and strong extractors.

\subsection{Registers and states}\label{sec:states}

A \emph{register} is a set $X$ that represents the classical states of a physical system. Note that we may have distinct registers with the same underlying set of states. We represent registers by uppercase Latin letters and classical states from the register by the corresponding lowercase letter. For registers $X_1$ and $X_2$, write the \emph{compound register} $X_1X_2=X_1\times X_2$, representing the states of both systems. A register $Y$ is a \emph{subregister} of $X$ if $X$ is a compound register with $Y$ as a factor. For a register~$X$, define the Hilbert space $\mc{H}_X$ as the $|X|$-dimensional space spanned by the orthonormal basis $\set*{\ket{x}}{x\in X}$ called the \emph{register basis}. The pure quantum states on $X$ are given by the unit vectors of $\mc{H}_X$, up to phase. We implicitly make use of the isomorphism $\mc{H}_{XY}\cong\mc{H}_X\otimes\mc{H}_Y$.

We write the set of \emph{linear operators} $\mc{H}_X\rightarrow\mc{H}_Y$ as $\mc{L}(X,Y)$, and if $X=Y$ as $\mc{L}(X)$; the set of \emph{positive semidefinite operators} on $\mc{H}_X$ as $\mc{P}(X)$, and when $X$ is evident, write $P\geq 0$ for $P\in\mc{P}(X)$; and the set of \emph{density operators} $\mc{D}(X)=\set*{\rho\in \mc{P}(X)}{\Tr(\rho)=1}$, representing the mixed quantum states. An operator $\rho\in\mc{P}(X)$ is a \emph{subnormalised state} if $\Tr(\rho)\leq 1$. The definitions below for mixed states extend directly to subnormalised states. Write $\Id_X\in\mc{L}(X)$ for the \emph{identity operator}, and $\id_X:\mc{L}(X)\rightarrow\mc{L}(X)$ for the \emph{identity channel}. For $\rho=\rho_{XY}\in \mc{D}(XY)$, write $\rho_X=\Tr_Y(\rho_{XY})$. A state $\rho\in\mc{D}(X)$ is \emph{classical} if it is diagonal in the register basis: it corresponds to a probability distribution on $X$. As a shorthand, write $[x]:=\ketbra{x}\in \mc{D}(X)$ to represent the density operator of a deterministic classical state.  A state $\rho\in \mc{D}(XY)$ is called \emph{classical-quantum} (cq) or classical on $X$ if it can be written $\rho=\sum_{x\in X}p_x[x]\otimes\rho^x_Y$ for some $\rho^x_Y\in \mc{D}(Y)$ and $p_x\in[0,1]$. By extension, we say a state $\rho\in \mc{D}(X_1\cdots X_mY_1\cdots Y_n)$ is $\mathrm{c}^m\mathrm{q}^n$ if it is classical on each $X_i$. We say a register $X$ is \emph{classical} to assume that every state we work with is classical on~it. We say that a state $\rho_X$ is \emph{supported} on $Y$ if $Y$ is a subregister of $X$.

We represent a probability distribution on a register $X$ by a function $\pi:X\rightarrow[0,1]$ such that $\sum_{x\in X}\pi(x)=1$. When the probability distribution is implicit, we write the probability of an event $\Omega\subseteq X$ as $\Pr\squ*{\Omega}=\sum_{x\in\Omega}\pi(x)$. For any $\C$-vector space $V$, we write the \emph{expectation value} with respect to the distribution as $\expec_{x\leftarrow\pi}f(x):=\sum_{x\in X}\pi(x)f(x)$. The classical state corresponding to~$\pi$ is written $\mu_\pi=\expec_{x\leftarrow\pi}[x]\in\mc{D}(x)$. For the uniform distribution, we write the expectation simply $\expec_{x\in X}$ and the state $\mu_X$. Abusing notation a bit, when we consider a random variable with values in a register $X$, we often refer to the variable as $X$ as well.

A linear map $\Phi:\mc{L}(X)\rightarrow\mc{L}(Y)$ is called \emph{completely positive} if for any register $Z$ and $P\in \mc{P}(ZX)$, $(\id_Z\otimes\Phi)(P)\geq0$. It is \emph{trace-preserving} if for any $P\in\mc{P}(X)$, $\Tr\parens*{\Phi(P)}=\Tr(P)$; and \emph{trace non-increasing} if $\Tr\parens*{\Phi(P)}\leq\Tr(P)$. The \emph{quantum channels} $X\rightarrow Y$ are the completely positive trace-preserving (CPTP) maps $\mc{L}(X)\rightarrow\mc{L}(Y)$ --- they represent the most general quantum operations. A \emph{positive operator-valued measurement} (POVM) is a map $P:S\rightarrow\mc{P}(X)$, where $S$ and $X$ are registers, such that $\sum_{s\in S}P(s)=\Id_X$; we write $P_s:=P(s)$. A POVM $P$ is a \emph{projector-valued measurement} (PVM) if $P_sP_{s'}=\delta_{s,s'}P_s$ for all $s,s'\in S$. We can associate various channels to a measurement. For a POVM $P:S\rightarrow\mc{P}(X)$, the \emph{destructive measurement channel} is $\Psi_P:\mc{L}(X)\rightarrow\mc{L}(S)$ defined as
\begin{align}
	\Psi_P(\rho)=\sum_{s\in S}\Tr(P_s\rho)[s],
\end{align}
representing the classical outcome of a measurement; and the \emph{nondestructive measurement channel} $\Phi_P:\mc{L}(X)\rightarrow\mc{L}(SX)$ defined as
\begin{align}
	\Phi_P(\rho)=\sum_{s\in S}[s]\otimes\sqrt{P_s}\rho\sqrt{P_s},
\end{align}
which represents both the classical outcome and the perturbed quantum state after the measurement. Evidently, $\Tr_X(\Phi_P(\rho))=\Psi_P(\rho)$. For a state $\rho_{XY}\in D(XY)$, write $\rho_{P(X)XY}=(\Phi_P\otimes\id_Y)(\rho_{XY})$. Similarly, if $\rho_{XY}$ is classical on $X$, for any function $f:X\rightarrow S$, we write $\rho_{f(X)XY}=\sum_{x\in X}p_x[f(x)x]\otimes\rho^x_Y$. For any cq state $\rho_{XY}$ and any event $\Omega\subseteq X$ --- which may be phrased as either a subset or as a relation --- write the \emph{partial state}
\begin{align}
	\rho_{\land\Omega}=\rho_{XY\land\Omega}=\sum_{x\in\Omega}p_x[x]\otimes\rho^x_{Y},
\end{align}
and the \emph{conditional state} $\rho_{|\Omega}=\frac{\rho_{\land\Omega}}{\Tr\rho_{\land\Omega}}$. If the event makes reference to a measurement, \emph{e.g.} $\Omega=(P(Y)=s)$, or a function evaluation, we assume that the measurement or evaluation is undertaken by the nondestructive channel, used to come up with the partial or conditional state, and then the result is forgotten by tracing out. This may perturb registers on which the state is non-classical, so we have to in particular assure ourselves that any two measurements in the same event are compatible.

\subsection{Finite vector spaces and subspace coset states} \label{sec:coset}

Consider the vector space of bit strings $V=\Z_2^n$ over the finite field $\Z_2$. The \emph{canonical basis} of $V$ is the set $E=\{e_1,\ldots,e_n\}$, where $e_i$ is the string that is $1$ at position $i$ and $0$ elsewhere. For any $u\in V$, we expand in the basis as $u=\sum_iu_ie_i$. The \emph{inner product} on $V\times V\rightarrow\Z_2$ is defined as $u\cdot v=\sum_{i}u_iv_i$. For any subspace $a\subseteq V$, the \emph{orthogonal subspace}
\begin{align}
	a^\perp=\set*{v\in V}{u\cdot v=0\;\forall\;u\in a}.
\end{align}
This satisfies $(a^\perp)^{\perp}=a$ and $\dim a+\dim a^\perp=\dim V=n$, but in general $\spn_{\Z_2}(a\cup a^\perp)=a+a^\perp\neq V$, for example $\{00,11\}^\perp=\{00,11\}$.

A subspace $a\subseteq V$ is called a \emph{register subspace} if it may be expressed as $a=\spn_{\Z_2}S$ for some $S\subseteq E$ \cite{JNV+20arxiv}. For a register subspace, we have that $a^\perp=\spn_{\Z_2}S^c$, and therefore that $a+a^\perp=V$. In this case, we get the canonical isomorphisms $V/a\cong a^\perp$ and $V/a^\perp\cong a$. We can easily express any register subspace by an indicator vector $\iota(a)\in V$ defined by $\iota(a)_i=1$ if and only if $e_i\in a$.

The space $V$ can be be seen as a register, giving the Hilbert space $\mc{H}_V\cong(\C^2)^{\otimes n}$.

\begin{definition}[\cite{CLLZ21,VZ21}]
	Let $a\subseteq V$ be a subspace. Given $t,t'\in V$, the \emph{subspace coset state}
	\begin{align}
		\ket{a_{t,t'}}=\frac{1}{\sqrt{|a|}}\sum_{u\in a}(-1)^{u\cdot t'}\ket{u+t}.
	\end{align}
\end{definition}

If $u\in t+a$ and $u'\in t'+a^{\perp}$, we have that $\ket{a_{u,u'}}$ is equal to $\ket{a_{t,t'}}$ up to global phase. To make use of this, for any subspace $a$, we fix a linear map $\Z_2^{n-\dim a}\rightarrow\Z_2^n$, $t\mapsto t_a$ such that $t\mapsto t_a+a$ is an isomorphism $\Z_2^{n-\dim a}\cong\Z_2^n/a$, and then take, for $t\in\Z_2^{n-\dim a}$ and $t'\in\Z_2^{\dim a}$, $\ket{a_{t,t'}}:=\ket{a_{t_a,t_{a}'{\hspace{-0.5mm}\scriptscriptstyle\perp}}}$. Then, the coset states $\set*{\ket{a_{t,t'}}}{t\in\Z_2^{n-\dim a},t'\in\Z_2^{\dim a}}$ are all distinct and form an orthonormal basis of $\mc{H}_V$.

If $a$ is a register subspace, there is a particularly good choice of map. For $a^\perp=\spn_{\Z_2}\{e_{i_1},\ldots,e_{i_m}\}$ with $i_1<i_2<\ldots<i_m$, we take $t_a=\sum_{j=1}^mt_je_{i_j}$. This allows us to write the subspace coset state in this case as a Wiesner state $\ket{a_{t,t'}}=\ket{x^\theta}$, where $x=t_a+t'_{a^\perp}$ and $\theta=\iota(a)$.

\subsection{Entropy and extractors} \label{sec:entropy}

Given a state $\rho_{XY}\in\mc{D}(XY)$, the \emph{conditional min-entropy} of $X$ given $Y$ is defined as
\begin{align}
	H_{\min}(X|Y)_{\rho}=-\lg\inf\set*{\Tr(\sigma_Y)}{\rho_{XY}\leq\Id_X\otimes\sigma_Y;\sigma_Y\in\mc{P}(Y)},
\end{align}
where $\lg$ is the base-two logarithm \cite{Ren05,Tom16}. Qualitatively, this represents the uncertainty on $X$, knowing $Y$. If $\rho$ is classical on $X$, this takes on a quantitative meaning: $2^{-H_{\min}(X|Y)}$ is the maximal probability of guessing $X$ when given the register $Y$. In the absence of side information, the conditional min-entropy becomes the min-entropy $H_{\min}(X)_{\rho}=-\lg\norm{\rho_X}$, where the norm here is the operator norm.

We will use strong extractors to go from a condition on the entropy to a near-independence of registers.

\begin{definition}[\cite{KT08}]
	Let $X,Y,Z$ be classical registers. A \emph{quantum-proof $(k,\varepsilon)$-strong extractor} is a function $e:X\times Y\rightarrow Z$ that satisfies the following property. Let $\rho_{XQ}$ be a subnormalised state, where $Q$ is a quantum register. If $H_{\min}(X|Q)\geq k$, then
	\begin{align}
		\norm*{\rho_{e(X,Y)YQ}-\mu_Z\otimes\mu_Y\otimes\rho_Q}_{\Tr}\leq\varepsilon,
	\end{align}
	where $\rho_{YXQ}=\mu_Y\otimes\rho_{XQ}$.
\end{definition}
Here, the norm is the trace norm $\norm{A}_{\Tr}=\frac{1}{2}\Tr\sqrt{A^\dag A}$. Due to \cite{DPVR12}, many constructions of extractors exist. Though we will tend to stay general, we give an example of their construction that is useful to keep in mind. For any $m,n\in\N$ and $\varepsilon>0$, there exists a quantum-proof $(8\lg(3m/2\varepsilon)+m,\varepsilon)$-strong extractor $e:\Z_2^n\times\Z_2^d\rightarrow\Z_2^m$, where $d\in O(\lg(m\sqrt{n}/\varepsilon)^2\lg m)$. For course, for this to be useful, we need that $k=8\lg(3m/2\varepsilon)+m<n$. Nevertheless, it is possible to achieve an exponentially small error $\varepsilon=\eta^m$ for any output length $m$ by taking $n>8\lg(3m/2)+(1+8\lg1/\eta)m\in O(m)$, though this requires the key length $d$ to be polynomial in $m$. This example absolutely defeats the original purpose of strong extractors, to extract a large amount of near-uniform randomness using a small seed, but is of great use in our cryptographic applications. 

\section{Novel Coset State Monogamy-of-Entanglement Property}\label{sec:moe}

In this section, we introduce and prove the MoE property that we make use of throughout the paper. In \cref{sec:weak-and-strong}, we recall the MoE properties of coset states that are already known. In \cref{sec:leaky}, we show our new leaky MoE property: the result is given in \cref{thm:stronger-monogamy}. Finally, in \cref{sec:uncertainty}, we show that this MoE property is equivalent to an entropic uncertainty relation, given as \cref{cor:entropic-moe}.

\subsection{Weak and strong MoE properties}\label{sec:weak-and-strong}

Let the register $V=\Z_2^n$ and $A$ be a set of subspaces of $\Z_2^n$ of dimension $n/2$: we take $A$ to either be the set of all register subspaces of dimension $n/2$ or all subspaces of dimension $n/2$. We consider the following monogamy-of-entanglement game, played between a referee Alice, who holds $V$, and two cooperating players, Bob and Charlie.

\begin{enumerate}[1.]
	\item Alice samples a uniformly random $a\in A$ and $t,t'\in\Z_2^{n/2}$. She prepares the state $\ket{a_{t,t'}}\in\mc{H}_V$ and sends it to Bob and Charlie.
	
	\item They act by an arbitrary channel $\Phi:\mc{L}(V)\rightarrow\mc{L}(BC)$ and then are isolated, so that Bob holds~$B$ and Charlie holds~$C$.
	
	\item Alice shares $a$ with Bob and Charlie, and they each make guesses of the pair $(t,t')$.
	
	\item Bob and Charlie win if their guesses are both correct.
\end{enumerate}

It was shown in \cite{CLLZ21} that the winning probability of this game is sub-exponentially small in $n$. This is called the \emph{weak monogamy-of-entanglement property of subspace coset states.}

There is also a \emph{strong monogamy-of-entanglement property}, conjectured in the same work, which constrains the winning probability of a related game. The difference here is that the winning condition is slackened: Bob needs only guess $t$ and Charlie needs only guess $t'$ correctly to win. It was shown in \cite{CV22} that the winning probability of this game is upper-bounded by $\sqrt{e}(\cos\tfrac{\pi}{8})^n$.

\subsection{The \gamename MoE property}\label{sec:leaky}

We exhibit an even stronger version of the MoE properties by showing that the same bound holds on a family of games that can only be easier to win. In the same setting as above, the game proceeds as follows:

\begin{enumerate}[1.]
	\item Alice samples a uniformly random $a\in A$ and $t,t'\in\Z_2^{n/2}$. She prepares the state $\ket{a_{t,t'}}\in\mc{H}_V$ and sends it to Bob and Charlie.
	
	\item They act by an arbitrary channel $\Phi:\mc{L}(V)\rightarrow\mc{L}(BC)$ and then are isolated, so that Bob holds~$B$ and Charlie holds~$C$.
	
	\item Alice shares $a$ with Bob and Charlie.
	
	\item Bob makes a guess $t_B$ of $t$, which is then given to Charlie; Charlie makes a guess $t'_C$ of $t'$.
	
	\item Bob and Charlie win if their guesses are both correct.
\end{enumerate}

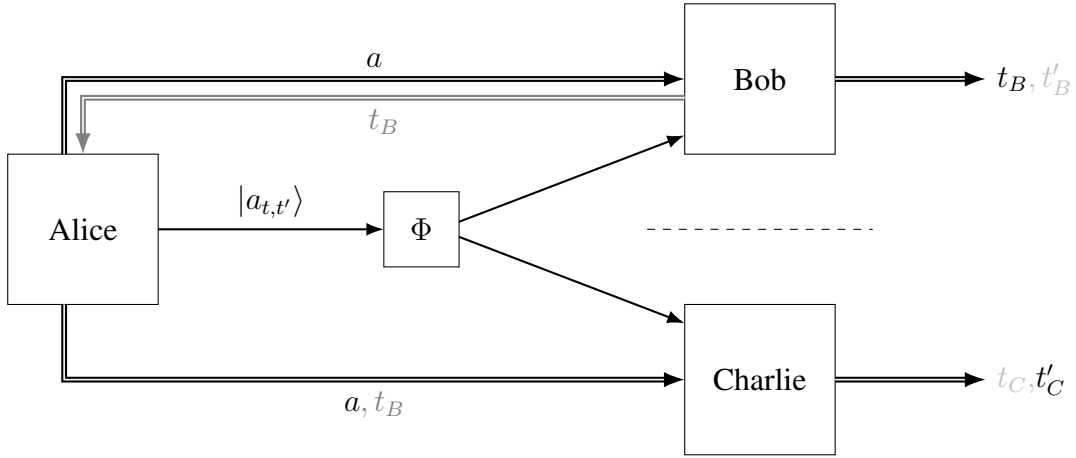
\begin{figure}[h!]
\centering
\begin{tikzpicture}
	\draw (-1,-1) rectangle (1,1) node[pos=0.5]{Alice};
	\draw[thick,-Latex] (1,0) -- (4,0) node[above, pos=0.5]{$\ket{a_{t,t'}}$};
	\draw (4,-0.5) rectangle (5,0.5) node[pos=0.5]{$\Phi$};
	\draw (8,1) rectangle (10,3) node[pos=0.5]{Bob};
	\draw (8,-1) rectangle (10,-3) node[pos=0.5]{Charlie};
	\draw[double, thick, -Latex] (-0.25,1) -- (-0.25,2) -- (8,2) node[above, pos=0.5]{$a$};
	\draw[double, thick, Latex-,gray] (0,1) -- (0,1.75) -- (8,1.75) node[below, pos=0.5]{$t_B$};
	\draw[double, thick, -Latex] (-0.25,-1) -- (-0.25,-2) -- (8,-2) node[below, pos=0.5]{$a{\color{gray},t_B}$};
	\draw[double,thick,-Latex] (10,2) -- (12,2) node[right]{$t_B{\color{lightgray},t'_B}$};
	\draw[double,thick,-Latex] (10,-2) -- (12,-2) node[right]{${\color{lightgray}t_C,}t'_C$};
	\draw[thick, -Latex] (5,0.1) -- (8,1.25);
	\draw[thick, -Latex] (5,-0.1) -- (8,-1.25);
	\draw[dashed] (7.5,0) -- (10.5,0);
\end{tikzpicture}
\caption{The subspace coset MoE games. The additional guesses Bob and Charlie need to make in the weak MoE game are given in {\color{lightgray}light gray}, and the additional interaction step in the \gamename MoE game is given in {\color{gray}dark gray}.}
\label{fig:MoE-games}
\end{figure}

We call this the \emph{$(n,A)$-\gamename monogamy-of-entanglement game}. The scenario is illustrated in \cref{fig:MoE-games}. An alternate but equivalent way to play the game, in order to bring it closer to the original form of an MoE game, is to have Alice provide Charlie with the correct value of $t$ rather than Bob's guess. The equivalence can be seen by noting that, in the original interpretation, only the cases when Bob's guess is correct are relevant to the computation of the winning probability. Next, we formalise the strategies and winning probability of this game.

\begin{definition}
	A quantum \emph{strategy} for the $(n,A)$-\gamename MoE game is a tuple of the form $\ttt{S}=\parens{B,C,\set{B^a}_{a\in A},\set{C^{a,t}}_{a\in A,t\in\Z_2^{n/2}},\Phi}$, where
	\begin{itemize}
		\item $B$ and $C$ are the registers representing Bob and Charlie's systems, respectively;
		
		\item $B^a:\Z_2^{n/2}\rightarrow\mc{P}(B)$ and $C^{a,t}:\Z_2^{n/2}\rightarrow\mc{P}(C)$ are POVMs, representing Bob and Charlie's measurements;
		
		\item $\Phi:\mc{L}(V)\rightarrow\mc{L}(BC)$ is a quantum channel, representing the splitting operation.
	\end{itemize}
	The \emph{winning probability} of a strategy $\ttt{S}$ is
	\begin{align}
		\mfk{w}_{n,A}(\ttt{S})&=\!\!\!\expec_{\substack{a\in A\\t,t'\in\Z_2^{n/2}}}\!\!\!\Tr\squ*{(B^a_t\otimes C^{a,t}_{t'})\Phi(\ketbra{a_{t,t'}})}.
	\end{align}
	The optimal winning probability of the $(n,A)$-\gamename MoE game is the supremum over all quantum strategies $\mfk{w}^\ast(n,A)=\sup_{\ttt{S}}\mfk{w}_{n,A}(\ttt{S})$.
\end{definition}

Now, we can formally express the \gamename MoE property.

\begin{theorem}\label{thm:stronger-monogamy}
	Let $n\in\N$ and $A$ be either the collection of register subspaces or the collection of all subspaces of dimension $n/2$ of $\Z_2^n$. Then,
	\begin{align}
		\mfk{w}^\ast(n,A)\leq\sqrt{e}(\cos\tfrac{\pi}{8})^n.
	\end{align}
\end{theorem}

First, we note that, as in \cite{TFKW13}, we need only consider strategies for the $(n,A)$-strong MoE game where the measurements $B^a$ and $C^{a,t}$ are projective, as any measurement may be made projective by dilating using Naimark's theorem. Next, we need an important lemma.

\begin{lemma}[Lemma 2 in \cite{TFKW13}]\label{lem:sum-bound}
	Let $P^s\in\mc{P}(H)$ for $s\in S$ be a collection of positive operators. For any set of mutually orthogonal permutations $\pi_s:S\rightarrow S$ (permutations such that $\pi_{s}\circ\pi_{t}^{-1}$ has a fixed point iff $s=t$) then
	\begin{align*}
	\norm[\Big]{\sum_{s\in S}P^s}&\leq\sum_{s\in S}\max_{t\in S}\norm*{\sqrt{P^t}\sqrt{P^{\pi_s(t)}}}.
	\end{align*}
\end{lemma}

The following technical lemma is  the final step of the proof of the theorem.

\begin{lemma}\label{lem:coset-overlap}
	For any $a,b\in A$, $\norm{P^aP^b}\leq\sqrt{\frac{|a\cap b|}{2^{n/2}}},$ where $P^a=\sum_{t,t'\in\Z_2^{n/2}}\ketbra{a_{t,t'}}\otimes B^a_t\otimes C^{a,t}_{t'}$.
\end{lemma}

\begin{proof}
	First, note that $P^a\leq\sum_{t,t'}\ketbra{a_{t,t'}}\otimes\Id_B\otimes C^{a,t}_{t'}$ and
	\begin{align}
		P^b\leq\sum_{u,u'}\ketbra{b_{u,u'}}\otimes B^b_{u}\otimes\Id_C=\sum_{u}\Pi_{b+u_b}\otimes B^b_{u}\otimes\Id_C,
	\end{align}
	where $\Pi_{b+u_b}=\sum_{v\in b+u_b}\ketbra{v}$ is the projector onto $b+u_b$. Then,
	\begin{align}
	\begin{split}
		\norm*{P^aP^b}&\leq\norm[\Big]{\sum_{t,t',u}\ketbra{a_{t,t'}}\Pi_{b+u_b}\otimes B^b_u\otimes C^{a,t}_{t'}}\\
		&=\max_{u\in\Z_2^{n/2}}\norm[\Big]{\sum_{t,t'}\ketbra{a_{t,t'}}\Pi_{b+u_b}\otimes C^{a,t}_{t'}},
	\end{split}
	\end{align}
	since the $B^b_u$ are orthogonal projectors. Next, by the $C^\ast$ identity,
	\begin{align}
		\norm[\Big]{\sum_{t,t'}\ketbra{a_{t,t'}}\Pi_{b+u_b}\otimes C^{a,t}_{t'}}=\norm[\Big]{\sum_{t,t'}\Pi_{b+u_b}\ketbra{a_{t,t'}}\Pi_{b+u_b}\otimes C^{a,t}_{t'}}^{1/2}.
	\end{align}
	Now, the terms in this sum are Hermitian with orthogonal supports, because $\Pi_{b+u_b}\ketbra{a_{t,t'}}\Pi_{b+u_b}$ provides the orthogonality for different values of $t$, and equal values of $t$, $C^{a,t}_{t'}$ provides it for different values of $t'$. Therefore, we can again decompose this norm as the maximum of the norms of each term. Putting this together, we get
	\begin{align}
	\begin{split}
		\norm*{P^aP^b}&\leq\max_{t,t',u\in\Z_2^{n/2}}\norm*{\Pi_{b+u_b}\ketbra{a_{t,t'}}\Pi_{b+u_b}}^{1/2}=\max_{t,t',u\in\Z_2^{n/2}}\sqrt{\braket{a_{t,t'}}{\Pi_{b+u_b}}{a_{t,t'}}},
	\end{split}
	\end{align}
	and we complete the proof by noting that
	\begin{align}
		\braket{a_{t,t'}}{\Pi_{b+u_b}}{a_{t,t'}}=\frac{1}{2^{n/2}}\sum_{v\in(a+t_a)\cap(b+u_b)}|(-1)^{t'_{a}{\scriptscriptstyle\perp}\cdot v}|^2\leq\frac{|a\cap b|}{2^{n/2}}.
	\end{align}
\end{proof}

Now, we can proceed to the proof of \cref{thm:stronger-monogamy}, which follows the method of the analogous proof in \cite{CV22}.

\begin{proof}[Proof of \cref{thm:stronger-monogamy}]
	First, for any strategy, we upper bound the winning probability by the norm of a related operator. Using the Choi-Jamio\l{}kowski representation $J(\Phi)=\frac{1}{2^n}\sum_{u,v\in\Z_2^n}\ketbra{u}{v}\otimes\Phi(\ketbra{u}{v})\in\mc{D}(VBC)$ of $\Phi$, we see that
	\begin{align}
	\begin{split}
		\mfk{w}_{n,A}(\ttt{S})&=\expec_{a\in A}\sum_{t,t'\in\Z_2^{n/2}}\Tr\squ*{(\ketbra{a_{t,t'}}\otimes B^a_t\otimes C^{a,t}_{t'})J(\Phi)}\\
		&\leq\norm[\Big]{\expec_{a\in A}\sum_{t,t'\in\Z_2^{n/2}}\ketbra{a_{t,t'}}\otimes B^a_t\otimes C^{a,t}_{t'}}.
	\end{split}
	\end{align}
	Using the notation of the previous lemma, this is $\mfk{w}_{n,A}(\ttt{S})\leq\norm*{\expec_aP^a}$. In the case that $A$ is the set of all subspaces of dimension $n/2$, we split the expectation into two: first we take the average over the bases $\beta$ of $\Z_2^n$, and then over the subspaces than can be spanned by that basis, that is
	\begin{align}
	\begin{split}
		\mfk{w}_{n,A}(\ttt{S})&\leq\norm[\Big]{\expec_{\beta\subseteq\Z_2^n\text{ basis}}\expec_{\gamma\subseteq\beta,|\gamma|=n/2}P^{\spn\gamma}}\leq\expec_{\beta\subseteq\Z_2^n\text{ basis}}\norm[\Big]{\expec_{\gamma\subseteq\beta,|\gamma|=n/2}P^{\spn\gamma}}.
	\end{split}
	\end{align}
	If A is the set of register subspaces, we don't need to take this step as we have $\mfk{w}_{n,A}(\ttt{S})\leq\norm{\expec_{\gamma\subseteq E,|\gamma|=n/2}P^{\spn\gamma}}$. In either case, we will complete the proof by fixing $\beta$ and showing that $\norm{\expec_{\gamma\subseteq\beta,|\gamma|=n/2}P^{\spn\gamma}}\leq\sqrt{e}(\cos\tfrac{\pi}{8})^n$. Let $S$ be the set of subsets of $\beta$ of cardinality $n/2$. There exists a family of orthogonal permutations $\pi_s:S\rightarrow S$ such that for each $k=0,\cdots,\frac{n}{2}$, the number of permutations such that $|\gamma\cap\pi_s(\gamma)|=\dim(\spn\gamma\cap\spn\pi_s(\gamma))=\frac{n}{2}-k$ for each $\gamma$ is $\binom{n/2}{k}^2$. Using \cref{lem:sum-bound} and then \cref{lem:coset-overlap}, we have, since $P^a$ is a projector,
	\begin{align}
	\begin{split}
		\norm[\Big]{\expec_{\gamma\in S}P^{\spn\gamma}}&\leq\expec_{s\in S}\max_{\gamma\in S}\norm{P^{\spn\gamma}P^{\spn\pi_s(\gamma)}}\\
		&\leq\expec_{s\in S}\max_{\gamma\in S}\sqrt{\frac{\abs*{\spn\gamma\cap\spn\pi_s(\gamma)}}{2^{n/2}}}\\
		&=\frac{1}{\binom{n}{n/2}}\sum_{k=0}^{n/2}\binom{n/2}{k}^2\sqrt{\frac{2^{n/2-k}}{2^{n/2}}}=\frac{1}{\binom{n}{n/2}}\sum_{k=0}^{n/2}\binom{n/2}{k}^22^{-k/2}.
	\end{split}
	\end{align}
	Using a result of \cite{CV22}, this is upper-bounded by $\sqrt{e}(\cos\tfrac{\pi}{8})^n$, finishing the proof.
\end{proof}

\subsection{A new type of entropic uncertainty relation}\label{sec:uncertainty}

We define a generalisation of the min-entropy that can be used to express MoE properties.

\begin{definition}
	Let $\rho$ be a state supported on not necessarily distinct classical registers $X_1,\ldots,X_n$ and quantum registers $A_1,\ldots,A_n$. For POVMs $M^i:X_i\rightarrow\mc{P}(A_i)$, write
	\begin{align}
		H_{\min}(X_1|M^1(A_1);\ldots;X_n|M^n(A_n))_\rho=-\lg\Tr\squ*{(\cdots(\rho_{\land(M^1(A_1)=X_1)})\cdots)_{\land(M^n(A_n)=X_n)}}.
	\end{align}
	Then, we define the \emph{sequential min-entropy} of $X_1,\ldots,X_n$ knowing $A_1,\ldots,A_n$ as
	\begin{align}
		H_{\min}(X_1|A_1;\ldots;X_n|A_n)_\rho=\inf_{M^1,\ldots,M^n\text{ POVMs}}H_{\min}(X_1|M^1(A_1);\ldots;X_n|M^n(A_n)).
	\end{align}
\end{definition}
Note that the sequential min-entropy is a generalisation of the conditional min-entropy in the sense that they are the same for $n=1$.

The winning probability of the $(n,A)$-\gamename MoE game may be phrased using this entropy. First, for registers $T=T'=\Z_2^{n/2}$ and $A$ representing either the register subspaces or all subspaces of $\Z_2^n$ of dimension $n/2$, Alice prepares $\rho_{ATT'}=\mu_A\otimes\mu_T\otimes\mu_{T'}$, and then copies $A$ and prepares coset states on $V=\Z_2^n$ accordingly to get
\begin{align}
	\rho_{AA'TT'V}=\expec_{a,t,t'}[aatt']\otimes\ketbra{a_{t,t'}}.
\end{align}
Bob and Charlie act with a channel $\Phi$, giving $\rho_{AA'TT'BC}=(\id_{AA'TT'}\otimes\Phi)(\rho_{AA'TT'V})$. In terms of the sequential min-entropy, the \gamename MoE property is the statement that
\begin{align}\label{eq:sequential-leaky-moe}
	H_{\min}(T|AB;T'|A'TC)_\rho\geq-\lg\mfk{w}^\ast(n,A)\geq(-\lg\cos\tfrac{\pi}{8})n-\tfrac{1}{2\ln 2}.
\end{align}
This expression follows directly from the definition. The only snarl is that, in general in the definition of the sequential min-entropy, Bob's measurement may not preserve $A$; and
similarly Charlie’s measurement may not preserve $A'T$. However, since these classical registers are not reused, only the diagonal blocks have any effect, and therefore, we may assume that the measurements are diagonal on the classical registers. As such, the infimum over the measurements is attained by those measurements that correspond to strategies. Note that any MoE game admits an entropic expression of this form.

To close off this section, we present a way to expand the sequential min-entropy as an entropic uncertainty relation.

\begin{proposition}\label{prop:sequential-to-entropic-uncertainty}
	Let $\rho$ be a state supported on classical registers $X,Y$ and quantum registers $A,B$. Then,
	\begin{align}
		H_{\min}(X|A;Y|B)_\rho=\inf_{M:X\rightarrow\mc{P}(A)\text{ POVM}}\parens*{H_{\min}(X|M(A))_\rho+H_{\min}(Y|B)_{\rho_{|(M(A)=X)}}}.
	\end{align}
\end{proposition}

Note the contrast between this entropic uncertainty relation and that found in \cite{TFKW13}. Most importantly, their relation considers the min-entropy of the same state on both terms, whereas ours uses different, albeit closely related, states. This avoids the shortcoming of their entropic uncertainty relation --- that the entropy can remain bounded for any dimension of Alice's space --- and thus allows us to make use of the full power of the MoE property in terms of an entropy.

\begin{proof}
	This follows immediately from the definition. We have
	\begin{align}
	\begin{split}
		H_{\min}(X|A;Y|B)&=\inf_{M,N}-\lg\Tr\squ*{(\rho_{\land(M(A)=X)})_{\land(N(B)=Y)}}\\
		&=\inf_{M,N}-\lg\Tr\squ*{\rho_{\land(M(A)=X)}}\Tr\squ*{(\rho_{|(M(A)=X)})_{\land(N(B)=Y)}}\\
		&=\inf_{M}\parens*{-\lg\Tr\squ*{\rho_{\land(M(A)=X)}}+\inf_N-\lg\Tr\squ*{(\rho_{|(M(A)=X)})_{\land(N(B)=Y)}}}\\
		&=\inf_{M}\parens*{H_{\min}(X|M(A))_\rho+H_{\min}(Y|B)_{\rho_{|(M(A)=X)}}}\qedhere
	\end{split}
	\end{align}
\end{proof}

Using the above proposition, we may express the \gamename MoE property as an entropic uncertainty relation.

\begin{corollary}[Leaky MoE entropic uncertainty relation]\label{cor:entropic-moe}
	For any measurement $M:T\rightarrow\mc{P}(AB)$ Bob makes in the leaky MoE game, we have
	\begin{align}
		H_{\min}(T|M(AB))_{\rho}+H_{\min}(T'|A'TC)_{\rho_{|(M(AB)=T)}}\geq(-\lg\cos\tfrac{\pi}{8})n-\tfrac{1}{2\ln2}.
	\end{align}
\end{corollary}

This follows immediately by combining \cref{thm:stronger-monogamy} with \cref{prop:sequential-to-entropic-uncertainty} via \cref{eq:sequential-leaky-moe}. This is the form of the bound that we make use of throughout the remainder of the paper. 

\section{Interactive Uncloneable Encryption}\label{sec:encryption}

In this section, we discuss our first application, introduced in \cref{sec:intro-encryption}. In \cref{sec:qecm-id}, we introduce the formalism used for interactive uncloneable encryption and discuss its security. In \cref{sec:encryption-construction}, we give a construction, given as \cref{prot:qecm-id}, and prove its security using the leaky MoE property of the previous section.

\subsection{QECMs with interactive decryption and their security} \label{sec:qecm-id}

We construct an uncloneable encryption scheme which requires only a communication assumption. That is, in order to decrypt a message, the sender Alice is required to have a short interaction with the receiver Bob. Note that, like uncloneable encryption, interactive uncloneable encryption does not assume an intended recipient, but once the interaction is started, only the party that initiated the interaction will be able to decrypt the message with high probability. First, in order to make sense of this interactive decryption, we extend the idea of a quantum encryption of classical messages of \cite{BL20}, by allowing the decryption to contain an interaction between the sender Alice and the receiver Bob. This allows for uncloneability via the \gamename MoE property, as it will permit Alice to check whether an eavesdropper has the ciphertext by checking whether Bob holds an uncorrelated piece of information. We present this formally.

\begin{definition}
	A \emph{quantum encryption of classical messages with interactive decryption (QECM-ID)} is a tuple $\ttt{Q}=\parens*{\ttt{Key},\ttt{Enc},\ttt{Dec}}$.
	\begin{itemize}
		\item $\ttt{Key}:\mc{D}(\{0\})\rightarrow\mc{D}(K)$ is the quantum channel representing the key-generation algorithm, where $K$ is the classical key register.
		\item $\ttt{Enc}:\mc{D}(KM)\rightarrow\mc{D}(KMC)$ is the quantum channel representing the encryption algorithm, where $M$ is the classical message register and $C$ is the quantum ciphertext register. $\ttt{Enc}$ preserves $KM$, \emph{i.e.} $\ttt{Enc}([km])=[km]\otimes\sigma^{km}_C$, where $\sigma^{km}_C$ is the quantum ciphertext.
		\item The decryption algorithm $\ttt{Dec}$ is an interaction between Alice and Bob that takes a state $\rho_{KMB}$ to $\rho_{KMF\hat{M}B'}$, where Alice holds $K$, $M$, and $F=\Z_2$ (a classical register that indicates whether Alice aborts (0) or accepts the decryption (1)); and Bob holds $\hat{M}$ (a classical register holding Bob's decryption of the message), and $B$ and $B'$  (additional quantum registers).
	\end{itemize}
	
	The scheme is \emph{$\varepsilon$-correct} if, for any classical state $\rho_M$, when Alice and Bob run $\ttt{Dec}$ as intended on $\rho_{KMC}=\ttt{Enc}(\ttt{Key}([0])\otimes\rho_M)$ for $B=C$, they get $\rho_{KMF\hat{M}}$ such that\footnote{We use this definition as it presents an operational way to simultaneously lower bound the probabilities of aborting and decrypting the correct message.} \begin{align}
		\norm{\rho_{M\hat{M}\land(F=1)}-\rho_{MM}}_{\Tr}\leq\varepsilon.
	\end{align}
\end{definition}

Note that this reduces to the original definition of a QECM if the decryption is a simple one-round interaction: Alice sends the key $k$ to Bob, who uses it to decrypt the ciphertext, and Alice always accepts the decryption. We extend the security properties of indistinguishable, uncloneable, and uncloneable-indistinguishable security of a QECM to this setting as well. Intuitively, the definitions are meant to replace the condition of Bob guessing correctly with Alice accepting the decryption.

First, we can describe the security properties by means of security games. The \emph{indistinguishable security game} is played by an adversary Bob against a challenger Alice.
\begin{enumerate}
	\item Bob prepares a cq state $\rho_{MS}$ and sends register $M$ to Alice, keeping hold of the side-information.
	
	\item Alice samples a bit $y$ uniformly at random. If $y=0$ she replaces $M$ with a fixed message~$m_0$; else she preserves $M$.
	
	\item Alice samples a key using $\ttt{Key}$ and encrypts the message. She then sends the ciphertext to Bob.
	
	\item Bob attempts to guess $y$. He wins if he guesses correctly.
\end{enumerate}
Indistinguishable security is achieved if the winning probability of this game is only slightly above~$\frac{1}{2}$. This is a standard property of encryption schemes.

Uncloneable security guarantees that, even if a colluding party decrypts, an eavesdropper can only guess the message as well as her side information allows. The \emph{uncloneable security game} is played by two cooperating adversaries Bob and Eve against a challenger Alice.
\begin{enumerate}
	\item Alice samples a message uniformly at random. She samples a key and encrypts the message. She sends the ciphertext to the adversaries.
	
	\item The adversaries split the state between them using a quantum channel, and then may no longer communicate.
	
	\item Alice and Bob decrypt with the interaction $\ttt{Dec}$, and Eve eavesdrops on their interactions.
	
	\item Eve attempts to guess the message. The adversaries win if Alice accepts the decryption ($f=1$) and Eve guesses correctly.
\end{enumerate}
Uncloneable security is achieved if the winning probability is only slightly above the probability of Alice accepting and Eve guessing the message given no information $\frac{\Pr[F=1]}{|M|}$.

Finally, uncloneable-indistinguishable security combines uncloneable and indistinguishable security: it guarantees that, even if a colluding party decrypts, an eavesdropper cannot distinguish between the encryptions of an intended message and a fixed message. The \emph{uncloneable-indistinguishable security game} is also played by two cooperating adversaries against a challenger.
\begin{enumerate}
	\item The adversaries prepare a cq state $\rho_{MS}$ and send register $M$ to Alice.
	
	\item Alice samples a bit $y$ uniformly at random. If $y=0$ she replaces $M$ with a fixed message~$m_0$; else she preserves $M$.
	
	\item Alice samples a key and encrypts the message. She sends the ciphertext to the adversaries.
	
	\item The adversaries split the state between them using a quantum channel, and then may no longer communicate.
		
	\item Alice and Bob decrypt with the interaction $\ttt{Dec}$, and Eve eavesdrops on their interactions.
	
	\item Eve tries to guess $y$. The adversaries win if Alice accepts the decryption and Eve guesses correctly.
\end{enumerate}
Uncloneable-indistinguishable security is achieved if the winning probability is only slightly above $\tfrac{1}{2}\Pr\squ{F=1}$, half the probability of accepting.

We now formalise the intuition of these security games in a way that is amenable to security proofs in the information-theoretic setting.

\begin{definition}
	Let $\ttt{Q}=\parens*{\ttt{Key},\ttt{Enc},\ttt{Dec}}$ be a QECM-ID. We say the scheme satisfies
	\begin{description}
		\item[$\varepsilon_1$-indistinguishable security] if
		\begin{align}
			\norm{\rho_{CS|(Y=0)}-\rho_{CS|(Y=1)}}_{\Tr}\leq\varepsilon_1,
		\end{align}
		for $\rho$ prepared as follows. Fix $m_0\in M$, and let $Y=\Z_2$ and $\rho_{MS}$ be any cq state. Alice prepares the state $\rho_{MSY}=\tfrac{1}{2}([m_0]\otimes\rho_S\otimes[0]+\rho_{MS}\otimes[1])$, then encrypts to get $\rho_{KMCSY}=(\ttt{Enc}\otimes\id_{SY})(\ttt{Key}([0])\otimes\rho_{MSY})$.
		
		\item[$\varepsilon_2$-uncloneable security] if
		\begin{align}
			\Pr\squ*{M=\check{M}\land F=1}_\rho\leq\frac{1}{|M|}\Pr\squ*{F=1}_\rho+\varepsilon_2,
		\end{align}
		for $\rho$ prepared as follows. Let $\rho_{M}=\mu_M$ be the maximally mixed state. Alice then encrypts~${\rho_{KMC}=\ttt{Enc}(\ttt{Key}([0])\otimes\rho_{M})}$ and an eavesdropper Eve acts with a quantum channel~${\Phi:\mc{L}(C)\rightarrow\mc{L}(BE)}$ to get $\rho_{KMBE}=(\id_{KM}\otimes\Phi)(\rho_{KMC})$. Then, after eavesdropping on all the interactions during $\ttt{Dec}$, Eve produces a guess $\check{M}$ of $M$.
		
		\item[$\varepsilon_3$-uncloneable-indistinguishable security] if
		\begin{align}
			\norm{\rho_{E'|(Y=0)\land(F=1)}-\rho_{E'|(Y=1)\land(F=1)}}_{\Tr}\leq\varepsilon_3,
		\end{align}
		for $\rho$ prepared as follows. Fix $m_0\in M$, and let $Y=\Z_2$ and $\rho_{MS}$ be any cq state. Alice prepares the state $\rho_{MSY}=\tfrac{1}{2}([m_0]\otimes\rho_S\otimes[0]+\rho_{MS}\otimes[1])$, then encrypts to get $\rho_{KMCSY}=(\ttt{Enc}\otimes\id_{SY})(\ttt{Key}([0])\otimes\rho_{MSY})$. Next, an eavesdropper Eve acts with a quantum channel $\Phi:\mc{L}(CS)\rightarrow\mc{L}(BE)$ to get $\rho_{KMBEY}=(\id_{KM}\otimes\Phi\otimes\id_Y)(\rho_{KMCSY})$ and after eavesdropping on all the interactions during $\ttt{Dec}$, Eve holds a register $E'$.
	\end{description}
\end{definition}

The security definitions are illustrated in \cref{fig:enc-defs}.

\begin{figure}[h!]
	\centering
	\begin{subfigure}{0.4\textwidth}
		\centering
		\begin{tikzpicture}
			\node (zero) at (0,0) {$0$};
			\node (key) at (2,0) {$K$};
			\node (message) at (2,-0.5) {$M$};
			\node (cipher) at (4,-1) {$C$};
			\node (mhat) at (6,-1) {$\hat{M}$};
			\node (message2) at (6,-0.5) {$M$};
			\node (flag) at (6,0.5) {$F$};

			\draw (zero) -- (key);
			\draw (3,-1) -- (cipher);
			\draw (cipher) -- (mhat);
			\draw (key) -- (5,0);
			\draw (message) -- (message2);
			\draw (flag) -- (5,0.5);
			\filldraw[fill=white] (1,0.2) node[above]{$\ttt{Key}$} -- (1.2,0.2) -- (1.2,-0.2) -- (0.8,-0.2) -- (0.8,0.2) -- cycle;
			\filldraw[fill=white] (3,0.1) node[above]{$\ttt{Enc}$} -- (3.2,0.1) -- (3.2,-1.1) -- (2.8,-1.1) -- (2.8,0.1) -- cycle;
			\filldraw[fill=white] (5,0.6) node[above]{$\ttt{Dec}$} -- (5.2,0.6) -- (5.2,-1.1) -- (4.8,-1.1) -- (4.8,0.6) -- cycle;
			\draw[dotted] (4.8,-0.75) -- (5.2,-0.75);
		\end{tikzpicture}
		\label{fig:enc-def-corr}
		\caption{Correctness}
	\end{subfigure}
	\hfill
	\begin{subfigure}{0.59\textwidth}
		\centering
		\begin{tikzpicture}[scale=0.9]
			\node (zero) at (0,0) {$0$};
			\node (key) at (2,0) {$K$};
			\node (message) at (2,-0.5) {$M$};
			\node at (1.3,-0.5) {$\mu_M\{$};
			\node (cipher) at (4,-1.25) {$C$};
			\node (bob) at (6,-1) {$B$};
			\node (eve) at (6,-1.5) {$E$};
			\node (mcheck) at (8,-1.5) {$\check{M}$};
			\node (message2) at (8,-0.5) {$M$};
			\node (flag) at (8,0.5) {$F$};

			\draw (zero) -- (key);
			\draw (3,-1.25) -- (cipher);
			\draw (cipher) -- (4.65,-1.25);
			\draw (5.35,-1) -- (bob);
			\draw (5.35,-1.5) -- (eve);
			\draw[-Latex] (bob) -- (7,-1) -- (7,-1.5);
			\draw (eve) -- (mcheck);
			\draw (key) -- (7,0);
			\draw (message) -- (message2);
			\draw (flag) -- (7,0.5);
			
			\filldraw[fill=white] (4.65,-0.9) rectangle (5.35,-1.6) node[pos=0.5]{$\Phi$};
			\filldraw[fill=white] (1,0.2) node[above]{$\ttt{Key}$} -- (1.2,0.2) -- (1.2,-0.2) -- (0.8,-0.2) -- (0.8,0.2) -- cycle;
			\filldraw[fill=white] (3,0.1) node[above]{$\ttt{Enc}$} -- (3.2,0.1) -- (3.2,-1.35) -- (2.8,-1.35) -- (2.8,0.1) -- cycle;
			\filldraw[fill=white] (7,0.6) node[above]{$\ttt{Dec}$} -- (7.2,0.6) -- (7.2,-1.1) -- (6.8,-1.1) -- (6.8,0.6) -- cycle;
			\draw[dotted] (6.8,-0.75) -- (7.2,-0.75);
		\end{tikzpicture}
		\label{fig:enc-def-unc}
		\caption{Uncloneability}
	\end{subfigure}

	\begin{subfigure}{0.4\textwidth}
		\centering
		\begin{tikzpicture}
			\node (zero) at (0,0) {$0$};
			\node (key) at (2,0) {$K$};
			\node (message) at (2,-1) {$M$};
			\node (sideinfo) at (2,-2) {$S$};
			\node (y) at (2,-0.5) {$Y$};
			\node at (1.6,-1.5) {$\Bigg\{$};
			\node at (1,-1.5) {$\rho_{MS}$};
			\node (cipher) at (5,-1.5) {$C$};
			\node (y2) at (5,-0.5) {$Y$};
			\node (side2) at (5,-2) {$S$};
			
			\draw (zero) -- (key);
			\draw (4,-1.5) -- (cipher);
			\draw (key) -- (4,0);
			\draw (message) -- (4,-1);
			\draw (y) -- (3.65,-0.5);
			\draw (4.35,-0.5) -- (y2);
			\draw (sideinfo) -- (side2);
			
			\filldraw[fill=white] (1,0.2) node[above]{$\ttt{Key}$} -- (1.2,0.2) -- (1.2,-0.2) -- (0.8,-0.2) -- (0.8,0.2) -- cycle;
			\filldraw[fill=white] (2.8,-0.4) rectangle (3.2,-1.1);
			\filldraw[fill=white] (4,0.1) node[above]{$\ttt{Enc}$} -- (4.2,0.1) -- (4.2,-1.6) -- (3.8,-1.6) -- (3.8,0.1) -- cycle;
		\end{tikzpicture}
		\label{fig:enc-def-indis}
		\caption{Indistinguishability}
	\end{subfigure}
	\hfill
	\begin{subfigure}{0.55\textwidth}
		\centering
		\begin{tikzpicture}
			\node (zero) at (0,0) {$0$};
			\node (key) at (2,0) {$K$};
			\node (message) at (2,-1) {$M$};
			\node (sideinfo) at (2,-2) {$S$};
			\node (y) at (2,-0.5) {$Y$};
			\node at (1.6,-1.5) {$\Bigg\{$};
			\node at (1,-1.5) {$\rho_{MS}$};
			\node (cipher) at (5,-1.5) {$C$};
			\node (y2) at (9,-0.5) {$Y$};
			\node (bob) at (7,-1.5) {$B$};
			\node (eve) at (7,-2) {$E$};
			\node (eve2) at (9,-2) {$E'$};
			\node (flag) at (9,0.5) {$F$};
					
			\draw (zero) -- (key);
			\draw (4,-1.5) -- (cipher);
			\draw (key) -- (8,0);
			\draw (message) -- (8,-1);
			\draw (y) -- (3.65,-0.5);
			\draw (4.35,-0.5) -- (7.65,-0.5);
			\draw (8.35,-0.5) -- (y2);
			\draw (sideinfo) -- (5.65,-2);
			\draw (cipher) -- (5.65,-1.5);
			\draw (5.65,-1.4) rectangle (6.35,-2.1) node[pos=0.5]{$\Phi$};
			\draw[-Latex] (6.35,-1.5) -- (bob) -- (8,-1.5) -- (8,-2);
			\draw (6.35,-2) -- (eve) -- (eve2);
			\draw (flag) -- (8,0.5);
			
			\filldraw[fill=white] (1,0.2) node[above]{$\ttt{Key}$} -- (1.2,0.2) -- (1.2,-0.2) -- (0.8,-0.2) -- (0.8,0.2) -- cycle;
			\filldraw[fill=white] (2.8,-0.4) rectangle (3.2,-1.1);
			\filldraw[fill=white] (4,0.1) node[above]{$\ttt{Enc}$} -- (4.2,0.1) -- (4.2,-1.6) -- (3.8,-1.6) -- (3.8,0.1) -- cycle;
			\filldraw[fill=white] (8,0.6) node[above]{$\ttt{Dec}$} -- (8.2,0.6) -- (8.2,-1.6) -- (7.8,-1.6) -- (7.8,0.6) -- cycle;
			\draw[dotted] (7.8,-1.25) -- (8.2,-1.25);
		\end{tikzpicture}
		\label{fig:enc-def-unc-indis}
		\caption{Uncloneability-indistinguishability}
	\end{subfigure}
	\caption{Schematics of the state constructions in the QECM-ID security definitions. Blocks represent operations, with interactions if they are split by a dotted line. Horizontal lines represent registers; they take part in the operations they touch. Vertical arrows represent eavesdropping.}
	\label{fig:enc-defs}
\end{figure}
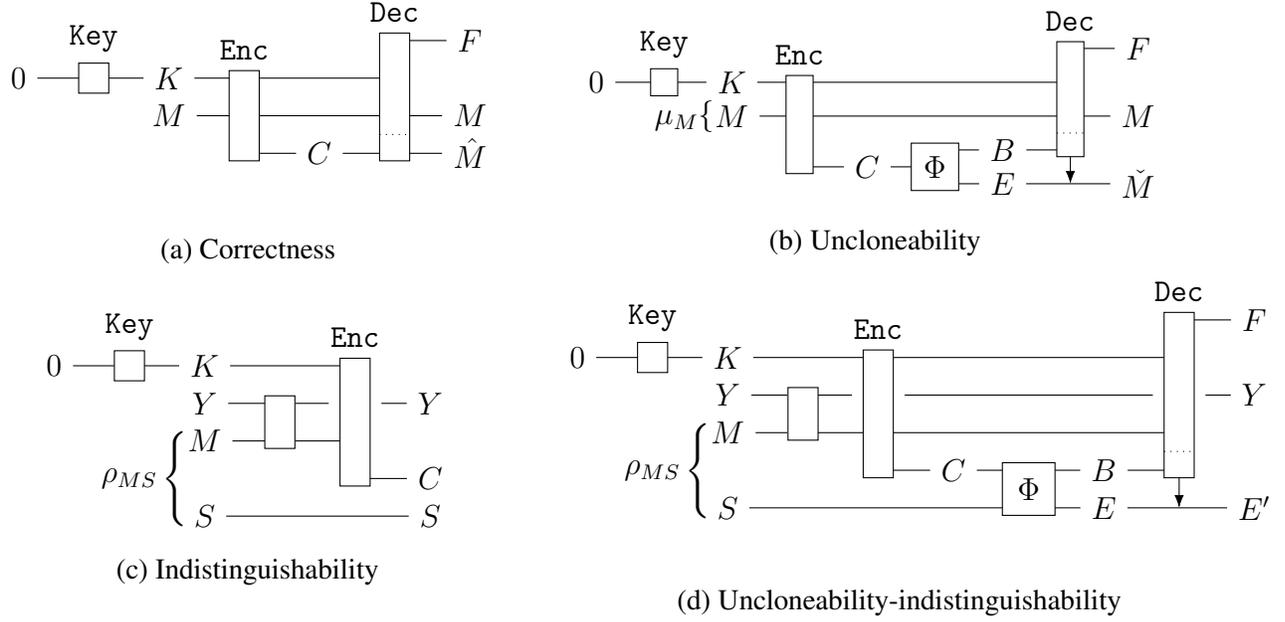

\subsection{General properties}

In this section, we show some relations on the uncloneable security properties for QECM-IDs, with the idea to generalise properties of classical encryption schemes. These extend and strengthen results known for QECMs.

First, we see that uncloneable security holds for non-uniform message distributions, generalising a property shown in \cite{BL20}.

\begin{lemma}
	Let $\ttt{Q}$ be an $\varepsilon$-uncloneable QECM-ID. Then, if the uncloneable security game is played with a classical state $\rho_M$ not necessarily uniform, the winning probability
	\begin{align}
		\Pr\squ*{M=\check{M}\land F=1}_\rho\leq2^{-H_{\min}(M)_\rho}\Pr[F=1]+|M|2^{-H_{\min}(M)_\rho}\varepsilon.
	\end{align}
\end{lemma}

\begin{proof}
	We relate this to the winning probability with $\rho_M=\mu_M$. In fact,
	\begin{align}
	\begin{split}
		\Pr\squ{M=\check{M}\land F=1}_\rho&=\sum_{m\in M}\Pr[M=m]\Pr\squ{\check{M}=m\land F=1|M=m}\\
		&\leq\max_m\Pr[M=m]\sum_m\Pr\squ{\check{M}=m\land F=1|M=m}\\
		&=|M|2^{-H_{\min}(M)_\rho}\Pr\squ{M=\check{M}\land F=1}_\mu\\
		&\leq2^{-H_{\min}(M)_\rho}\Pr[F=1]+|M|2^{-H_{\min}(M)_\rho}\varepsilon
	\end{split}
	\end{align}
\end{proof}

Next, we find an equivalence, up to scalar multiple of the parameters, between the uncloneable and uncloneable-indistinguishable security properties. One direction, uncloneable security implying uncloneable-indistinguishable security, generalises a similar property shown for QECMs in \cite{BL20}, while the other direction is new, and remains an open question for QECMs in the information-theoretic setting. The equivalence of these security properties is similar to the equivalence of semantic security and indistinguishability in classical encryption.

\begin{theorem}
	Let $\ttt{Q}$ be a perfectly indistinguishable QECM-ID.
	\begin{itemize}
		\item If $\ttt{Q}$ is $\varepsilon$-uncloneable secure then it is $|M|\varepsilon$-uncloneable-indistinguishable secure.
		
		\item If $\ttt{Q}$ is $\varepsilon$-uncloneable-indistinguishable secure then it is $\varepsilon$-uncloneable secure.
	\end{itemize}
\end{theorem}

Note that this theorem means that, outside of some pathological cases, it is only necessary to show either uncloneable and uncloneable-indistinguishable security for QECM-IDs, not both. However, we nevertheless show both in the following section, as it allows us to work out better parameters.

\begin{proof}\hphantom{}
\begin{itemize}
	\item We proceed by contrapositive. Suppose there exists an attack for the uncloneable-indistinguishable security game that wins with advantage greater than $|M|\varepsilon$. An important observation we make to help simplify the proof is that we may always assume that $\rho_{MS}=[m_1]$ for some message $m_1\in M$ \cite{KT22arxiv}. This is because the trace norm is convex, so
	\begin{align}
		\norm*{\rho_{E'|(Y=0)\land(F=1)}-\rho_{E'|(Y=1)\land(F=1)}}_{\Tr}\leq\sum_{m\in M}p_m\norm*{\rho^m_{E'|(Y=0)\land(F=1)}-\rho^m_{E'|(Y=1)\land(F=1)}}_{\Tr},
	\end{align}
	and thus we can take $m_1$ to be the value whose term in this convex combination is maximal. Finally, we can remove the side information by redefining the splitting channel $\Phi'(\sigma)=\Phi(\sigma\otimes\rho^{m_1}_{S})$.
	
	With such an attack, we construct an attack against the uncloneable security game. The splitting operation and Bob act in the same way. To attempt to guess the message, Charlie makes the measurement that optimally distinguishes the cases $y=0$ and $y=1$, and guess $m_0$ or $m_1$, respectively. Then, the guessing probability
	\begin{align}
		\begin{split}
			\Pr\squ*{M=\check{M}\land F=1}&=\Pr\squ*{M=\check{M}\land F=1\land M\notin\{m_0,m_1\}}\\
			&+\Pr\squ*{M\in\{m_0,m_1\}}\Pr\squ*{M=\check{M}\land F=1|M\in\{m_0,m_1\}}\\
			&=\frac{2}{|M|}\Pr\squ*{M=\check{M}\land F=1|M\in\{m_0,m_1\}}
		\end{split}
	\end{align}
	Since $\Pr\squ*{M=\check{M}\land F=1|M\in\{m_0,m_1\}}$ is the probability of distinguishing messages $m_0$ and $m_1$, we have by hypothesis that this is greater than $\frac{\Pr[F=1|M\in\{m_0,m_1\}]+|M|\varepsilon}{2}$. Finally, as $\ttt{Q}$ is perfectly indistinguishable, $\Pr[F=1|M\in\{m_0,m_1\}]=\Pr[F=1]$ --- otherwise Bob could distinguish the messages without access to the key. Putting this together,
	\begin{align}
		\Pr\squ*{M=\check{M}\land F=1}>\frac{\Pr[F=1]}{|M|}+\varepsilon.
	\end{align}

	\item Let $\rho_{ME'\land(F=1)}=\expec_{m\in M}[m]\otimes\rho^m_{E'\land(F=1)}$ be the final state in the uncloneable security game. Since we have by hypothesis that $\ttt{Q}$ is uncloneable-indistinguishable secure, $\norm{\rho^{m_0}_{E'\land(F=1)}-\rho^m_{E'\land(F=1)}}_{\Tr}\leq\varepsilon$ for all $m\in M$. Setting the state $\tau_{ME'\land(F=1)}=\mu_M\otimes\rho^{m_0}_{E'\land(F=1)}$, we have that
	\begin{align}
		\norm{\tau_{ME'\land(F=1)}-\rho_{ME'\land(F=1)}}_{\Tr}=\expec_{m\in M}\norm{\rho^{m_0}_{E'\land(F=1)}-\rho^m_{E'\land(F=1)}}_{\Tr}\leq\varepsilon.
	\end{align}
	Because the registers $M$ and $E'$ are independent on $\tau$, the guessing probability $\Pr[M=\check{M}\land F=1]_\tau\leq\frac{\Pr[F=1]_\tau}{|M|}$. Finally, because $\tau$ is only $\varepsilon$ away from $\rho$ in trace norm and $\Pr[F=1]_\tau=\Pr[F=1|M=m_0]_\rho=\Pr[F=1]_\rho$ by perfect indistinguishability, we get that $\Pr[M=\check{M}\land F=1]_\rho\leq\frac{\Pr[F=1]_\rho}{|M|}+\varepsilon$. \qedhere
\end{itemize}
\end{proof}

\subsection{Instantiation and security proofs}\label{sec:encryption-construction}

Now, we give a construction of a QECM-ID. Let $e:\Z_2^{n/2}\times R\rightarrow\Z_2^\ell$ be a quantum-proof $(\kappa,\varepsilon)$-strong extractor and let $A$ be the set of all subspaces of $V=\Z_2^n$ of dimension $n/2$.

\begin{mdframed}
\begin{protocol}[Coset state QECM-ID]\label{prot:qecm-id}\hphantom{}

\begin{description}
	\item[Key generation] Let $T=T'=\Z_2^{n/2}$ and $H=\Z_2^\ell$ and take $K=ATT'RH$. The channel
	\begin{align}
		\ttt{Key}([0])=\expec_{a,t,t',r,h}[att'rh].
	\end{align}
	
	\item[Encryption] Let $M=\bar{M}=\Z_2^\ell$ and $C=\bar{M}V$. Take
	\begin{align}
		\ttt{Enc}([att'rh]\otimes[m])=[att'rh]\otimes[m]\otimes[m+e(t',r)+h]\otimes\ketbra{a_{t,t'}}.
	\end{align}
	
	\item[Decryption] $\ttt{Dec}$ proceeds as follows. First, Alice sends $a$ to Bob. Then, Bob measures $V$ in the coset state basis to get measurements $\hat{t},\hat{t}'$ of $t,t'$. Bob sends $\hat{t}$ to Alice: if $\hat{t}=t$, Alice sets $f=1$, else she sets $f=0$ and aborts. Alice sends $r$ and $h$ to Bob. Bob computes $\hat{m}=\bar{m}+e(\hat{t}',r)+h$.
\end{description}
\end{protocol}
\end{mdframed}

\begin{proposition}
    \cref{prot:qecm-id} is perfectly correct, \emph{i.e.} $0$-correct.
\end{proposition}

\begin{proof}
    First, writing $\rho_M=\sum_{m}p_m[m]$,
    \begin{align}
        \rho_{KMC}=\rho_{ATT'RHM\bar{M}V}=\expec_{a,t,t',r,h}\sum_mp_m[att'rh]\otimes[m]\otimes[m+e(t',r)+h]\otimes\ketbra{a_{t,t'}}.
    \end{align}
    To begin the decryption, Bob measures in the coset state basis and gets
    \begin{align}
        \rho_{ATT'RHM\bar{M}\hat{T}\hat{T}'}=\expec_{a,t,t',r,h}\sum_mp_m[att'rh]\otimes[m]\otimes[m+e(t',r)+h]\otimes[tt'].
    \end{align}
    Sending $\hat{t}=t$ to Alice, she always sets $F=1$, and then gives $r$ and $h$ to Bob. Then, the state become
    \begin{align}
        \rho_{ATT'RHMF\bar{M}\hat{T}'}=\expec_{a,t,t',r,h}\sum_mp_m[att'rh]\otimes[m]\otimes[1]\otimes[m+e(t',r)+h]\otimes[t'].
    \end{align}
    Finally, Bob computes $\hat{m}=\bar{m}+e(\hat{t}',r)+h=m$, getting
    \begin{align}
        \rho_{KMF\hat{M}}=\expec_{a,t,t',r,h}\sum_mp_m[att'rh]\otimes[m]\otimes[1]\otimes[m].
    \end{align}
    Thus, $\rho_{M\hat{M}\land(F=1)}=\sum_mp_m[m]\otimes[m]=\rho_{MM}$.
\end{proof}
\newpage
\begin{proposition}
    \cref{prot:qecm-id} is perfectly indistinguishable.
\end{proposition}

\begin{proof}
    Writing $\rho_{MS}=\sum_mp_m[m]\otimes\rho^m_S$, we see that
    \begin{align}
    \begin{split}
        \rho_{KMCSY}&=\frac{1}{2}\parens*{\ttt{Enc}(\ttt{Key}(0)\otimes[m_0])\otimes\rho_S\otimes[0]+(\ttt{Enc}\otimes\id_S)(\ttt{Key}(0)\otimes\rho_{MS})\otimes[1]}\\
        &=\frac{1}{2}\sum_m\ttt{Enc}(\ttt{Key}(0)\otimes[m])\otimes\parens*{\delta_{m,m_0}\rho_S\otimes[0]+p_m\rho^m_S\otimes[1]}\\
        &=\frac{1}{2}\sum_m\expec_{a,t,t',r,h}[att'rhm]\otimes[m+e(t',r)+h]\otimes\ketbra{a_{t,t'}}\otimes\parens*{\delta_{m,m_0}\rho_S\otimes[0]+p_m\rho^m_S\otimes[1]}.
    \end{split}
    \end{align}
    Hence,
    \begin{align}
    \begin{split}
        \rho_{CSY}&=\frac{1}{2}\sum_m\expec_{a,t,t',r,h}[m+e(t',r)+h]\otimes\ketbra{a_{t,t'}}\otimes\parens*{\delta_{m,m_0}\rho_S\otimes[0]+p_m\rho^m_S\otimes[1]}\\
        &=\frac{1}{2}\expec_{a,t,t'}\mu_{\bar{M}}\otimes\ketbra{a_{t,t'}}\otimes\sum_m\parens*{\delta_{m,m_0}\rho_S\otimes[0]+p_m\rho^m_S\otimes[1]}\\
        &=\frac{1}{2}\mu_{C}\otimes\parens*{\rho_S\otimes[0]+\rho_S\otimes[1]}=\mu_C\otimes\rho_S\otimes\mu_Y.
    \end{split}
    \end{align}
    Thus, $\rho_{CS|(Y=0)}=\rho_{CS|(Y=1)}$.
\end{proof}

\begin{theorem}
    Suppose $\kappa\geq \frac{-\lg\cos\frac{\pi}{8}}{2}n-\frac{1}{4\ln 2}$. Then, \cref{prot:qecm-id} is $\max\{\varepsilon,e^{1/4}(\cos\tfrac{\pi}{8})^{n/2}\}$-uncloneable.
\end{theorem}

\begin{proof}
	We have the state before decryption
	\begin{align}
		\rho_{ATT'RHMBE}=\expec_{a,t,t',r,h,m}[att'rh]\otimes[m]\otimes\Phi([m+e(t',r)+h]\otimes\ketbra{a_{t,t'}}).
	\end{align}
	To begin the decryption, Alice shares $a$, and Bob makes a measurement $N$ on $B$ to determine a guess $\hat{t}$ of $t$. Fix $\bar{m}\in\bar{M}$. Then, taking $\sigma\mapsto\Phi([\bar{m}]\otimes\sigma)$ to be the cloning channel in the \gamename MoE game, we get by the \gamename MoE property that $H_{\min}(T|AB;T'|A'TE)_{\rho_{|(\bar{M}=\bar{m})}}\geq(-\lg\cos\tfrac{\pi}{8})n-\frac{1}{2\ln 2}$, where $A'$ is a copy of $A$. Thus, we must have either $H_{\min}(T|N(AB))_{\rho_{|(\bar{M}=\bar{m})}}\geq\frac{-\lg\cos\frac{\pi}{8}}{2}n-\frac{1}{4\ln 2}$ or $H_{\min}(T'|A'TE)_{\rho_{|(N(AB)=T\land \bar{M}=\bar{m})}}\geq \frac{-\lg\cos\frac{\pi}{8}}{2}n-\frac{1}{4\ln 2}$. In the former case, as $AB$ is the register Bob has access to by that point, we have
	\begin{align}
		\Pr[F=1]=\Pr[\hat{T}=T]=\Pr[N(AB)=T]\leq e^{1/4}(\cos\tfrac{\pi}{8})^{n/2}.
	\end{align}
	In the latter case, we have by hypothesis and the strong extractor property,
	\begin{align}
		\begin{split}
			&\norm{\rho_{e(T',R)RA'TE|(F=1\land \bar{M}=\bar{m})}-\mu_{\tilde{M}}\otimes\mu_R\otimes\rho_{A'TE|(F=1\land M=m)}}_{\Tr}\\
			&=\norm{\rho_{e(T',R)RA'TE|(N(AB)=T\land \bar{M}=\bar{m})}-\mu_{\tilde{M}}\otimes\mu_R\otimes\rho_{A'TE|(N(AB)=T\land\bar{M}=\bar{m})}}_{\Tr}\leq\varepsilon,
		\end{split}
	\end{align}
	where $\tilde{M}=\Z_2^\ell$ is the register containing $e(T',R)$. Combining the two cases,
	\begin{align}
		\begin{split}
			&\norm{\rho_{e(T',R)RA'TE\land(F=1)|(\bar{M}=\bar{m})}-\mu_{\tilde{M}}\otimes\mu_R\otimes\rho_{A'TE\land(F=1)|(\bar{M}=\bar{m})}}_{\Tr}\\
			&=\Pr[F=1]_\rho\norm{\rho_{e(T',R)RA'TE|(F=1\land \bar{M}=\bar{m})}-\mu_{\tilde{M}}\otimes\mu_R\otimes\rho_{A'TE|(F=1\land\bar{M}=\bar{m})}}_{\Tr}\\
			&\leq\varepsilon^\ast,
		\end{split}
	\end{align}
	where we set $\varepsilon^\ast=\max\{\varepsilon,e^{1/4}(\cos\tfrac{\pi}{8})^{n/2}\}$. This implies that, as $m$ and $\bar{m}=m+e(t',r)+h$ are uniformly distributed and independent,
	\begin{align}
		\begin{split}
			\rho_{M\bar{M}e(T',R)RA'TE\land(F=1)}&=\expec_{m,\bar{m}}[m\bar{m}]\otimes\rho_{e(T',R)RA'TE\land(F=1)|(\bar{M}=\bar{m})}\\
			&\approx_{\varepsilon^\ast}\expec_{m,\bar{m}}[m\bar{m}]\otimes\mu_{\tilde{M}}\otimes\mu_R\otimes\rho_{A'TE\land(F=1)|(\bar{M}=\bar{m})},
		\end{split}
	\end{align}
	hence $\norm{\rho_{\tilde{M}RM\bar{M}A'TE\land(F=1)}-\mu_{\tilde{M}RM}\otimes\rho_{\bar{M}A'TE\land(F=1)}}_{\Tr}\leq\varepsilon^\ast$. Supposing $f=1$, the decryption continues and Eve also gets $h=m+\bar{m}+\tilde{m}$ and tries to guess $m$. As classical computations are CPTP maps, we see that
	\begin{align}
		\begin{split}
			&\norm{\rho_{R\tilde{M}M\bar{M}A'TE\land(F=1)}-\mu_{R\tilde{M}M}\otimes\rho_{\bar{M}A'TE\land(F=1)}}_{\Tr}\\
			&\geq\norm{\rho_{R\tilde{M}M(M+\bar{M}+\tilde{M})\bar{M}A'TE\land(F=1)}-\mu_R\otimes\sigma_{\tilde{M}M(M+\bar{M}+\tilde{M})\bar{M}A'TE\land(F=1)}}_{\Tr}\\
			&\geq\norm{\rho_{RM(M+\bar{M}+\tilde{M})A'TE\land(F=1)}-\mu_R\otimes\sigma_{M(M+\bar{M}+\tilde{M})A'TE\land(F=1)}}_{\Tr},
		\end{split}
	\end{align}
	where $\sigma_{\tilde{M}M\bar{M}A'TE\land(F=1)}=\mu_{\tilde{M}M}\otimes\rho_{MA'TE\land(F=1)}$, so
	\begin{align}
		\begin{split}
			\sigma_{M(M+\bar{M}+\tilde{M})A'TE\land(F=1)}&=\expec_{m,\bar{m},\tilde{m}}[m(m+\bar{m}+\tilde{m})]\otimes\rho_{A'TE\land(F=1)|(\bar{M}=\bar{m})}\\
			&=\mu_{MH}\otimes\rho_{A'TE\land(F=1)}.
		\end{split}
	\end{align}
	
	During the decryption, all the information Eve receives is contained in $E'=RHA'TE$. Let the subnormalised state $\tau_{ME'}=\mu_{MRH}\otimes\rho_{A'TE\land(F=1)}$. By the above, we have that $\norm{\rho_{ME'\land(F=1)}-\tau_{ME'}}_{\Tr}\leq\varepsilon^\ast$. As such, if the shared state were $\tau$, $M$ is independent from $E'$, and therefore $\Pr\squ*{M=\check{M}\land F=1}_\tau\leq\frac{\Tr\tau}{|M|}=\frac{\Pr[F=1]_\rho}{|M|}$. This implies that the probability of guessing $M$ given $E'$ of $\rho_{ME'\land(F=1)}$ is at most
	\begin{align}
		\Pr[M=\check{M}\land F=1]_\rho\leq\Pr\squ*{M=\check{M}\land F=1}_\tau+\varepsilon^\ast\leq\frac{\Pr[F=1]_\rho}{|M|}+\varepsilon^\ast,
	\end{align}
	as wanted.
\end{proof}

\begin{theorem}
    Suppose $\kappa\geq \frac{-\lg\cos\frac{\pi}{8}}{2}n-\frac{1}{4\ln 2}$. Then, \cref{prot:qecm-id} is $\max\{2\varepsilon,2e^{1/4}(\cos\tfrac{\pi}{8})^{n/2}\}$ -indistinguishable-uncloneable.
\end{theorem}

\begin{proof}
    With $\rho_{MS}=\sum_mp_m[m]\otimes\rho^m_S$, we have again
    \begin{align}
        &\rho_{KMCSY}=\rho_{ATT'RHM\bar{M}VSY}\\
        &=\frac{1}{2}\sum_m\expec_{a,t,t',r,h}[att'rhm]\otimes[m+e(t',r)+h]\otimes\ketbra{a_{t,t'}}\otimes\parens*{\delta_{m,m_0}\rho_S\otimes[0]+p_m\rho^m_S\otimes[1]},\nonumber
    \end{align}
    so given the cloning attack $\Phi:\mc{L}(\bar{M}VS)\rightarrow\mc{L}(BE)$, the state before decryption is
    \begin{align}
        \rho_{KMBEY}=\frac{1}{2}\sum_m\!\!\expec_{a,t,t',r,h}\!\!\![att'rhm]\otimes\big(&\delta_{m,m_0}\Phi([m+e(t',r)+h]\otimes\ketbra{a_{t,t'}}\otimes\rho_S)\otimes[0]\\[-1em]
        &+p_m\Phi([m+e(t',r)+h]\otimes\ketbra{a_{t,t'}}\otimes\rho^m_S)\otimes[1]\big).\nonumber
    \end{align}
    On $\rho_{|(Y=0\land \bar{M}=\bar{m})}$, the cloning attack is $\sigma\mapsto\Phi([\bar{m}]\otimes\sigma\otimes\rho_S)$, so we have
    \begin{align}
    	H_{\min}(T|AB;T'|A'TE)_{\rho_{|(Y=0\land \bar{M}=\bar{m})}}\geq(-\lg\cos\tfrac{\pi}{8})n-\tfrac{1}{2\ln 2},
    \end{align}
	where $A'$ is a copy of $A$, and hence as above
    \begin{align}
        \norm{\rho_{e(T',R)RA'TE|(Y=0\land\bar{M}=\bar{m})\land(F=1)}-\mu_{\tilde{M}}\otimes\mu_R\otimes\rho_{A'TE|(Y=0\land\bar{M}=\bar{m})\land(F=1)}}_{\Tr}\leq\varepsilon^\ast,
    \end{align}
    and then, in order to include $M$ and $\bar{M}$,
    \begin{align}
    \begin{split}
        \rho_{Me(T',R)RA'TE\bar{M}|(Y=0)\land(F=1)}&=\expec_{\bar{m}}[m_0]\otimes\rho_{e(T',R)RA'TE|(Y=0\land\bar{M}=\bar{m})\land(F=1)}\otimes[\bar{m}]\\
        &\approx_{\varepsilon^\ast}\expec_{\bar{m}}[m_0]\otimes\mu_{\tilde{M}R}\otimes\rho_{A'TE|(Y=0\land\bar{M}=\bar{m})\land(F=1)}\otimes[\bar{m}]\\
        &=[m_0]\otimes\mu_{\tilde{M}R}\otimes\rho_{A'TE\bar{M}|(Y=0)\land(F=1)}.
    \end{split}
    \end{align}
	As $H=\tilde{M}+M+\bar{M}$, this gives that
	\begin{align*}
		\norm{\rho_{MRHA'TE|(Y=0)\land(F=1)}-[m_0]\otimes\mu_{RH}\otimes\rho_{A'TE|(Y=0)\land(F=1)}}_{\Tr}\leq\varepsilon^\ast
	\end{align*}
    In the same way, we get that on $\rho_{|(Y=1\land M=m\land\bar{M}=\bar{m})}$, the cloning attack is $\sigma\mapsto\Phi([\bar{m}]\otimes\sigma\otimes\rho_S^m)$, so the entropy $H_{\min}(T|AB;T'|A'TE)_{\rho_{|(Y=1\land M=m\land\bar{M}=\bar{m})}}\geq(-\lg\cos\tfrac{\pi}{8})n-\tfrac{1}{2\ln 2}$, and hence as above
    \begin{align}
        \norm{\rho_{e(T',R)RATE|(Y=1\land M=m)\land(F=1)}-\mu_{\tilde{M}}\otimes\mu_R\otimes\rho_{ATE|(Y=1\land M=m)\land(F=1)}}_{\Tr}\leq\varepsilon^\ast.
    \end{align}
    To include $M$ and $\bar{M}$,
    \begin{align}
    \begin{split}
        \rho_{e(T',R)RA'TEM\bar{M}|(Y=1)\land(F=1)}&=\sum_m\expec_{\bar{m}}p_m\rho_{e(T',R)RA'TE|(Y=1\land M=m\land \bar{M}=\bar{m})\land(F=1)}\otimes[m\bar{m}]\\
        &\approx_{\varepsilon^\ast}\sum_m\expec_{\bar{m}}p_m\mu_{\tilde{M}R}\otimes\rho_{A'TE|(Y=1\land M=m)\land(F=1)}\otimes[m\bar{m}]\\
        &=\mu_{\tilde{M}R}\otimes\rho_{A'TEM\bar{M}|(Y=1)\land(F=1)}.
    \end{split}
    \end{align}
    This gives that, again using $H=\tilde{M}+M+\bar{M}$,
    \begin{align}
        \norm{\rho_{RHMA'TE|(Y=1)\land(F=1)}-\mu_{RH}\otimes\rho_{MA'TE|(Y=1)\land(F=1)}}_{\Tr}\leq\varepsilon^\ast.
    \end{align}
    As $E'=RHA'TE$, this implies that
    \begin{align}
    \begin{split}
        &\norm{\rho_{E'|(Y=0)\land(F=1)}-\rho_{E'|(Y=1)\land(F=1)}}_{\Tr}\\
        &\leq\norm{\mu_{RH}\otimes\rho_{A'TE|(Y=0)\land(F=1)}-\mu_{RH}\otimes\rho_{A'TE|(Y=1)\land(F=1)}}_{\Tr}+2\varepsilon^\ast\\
        &=\norm{\rho_{A'TE|(Y=0)\land(F=1)}-\rho_{A'TE|(Y=1)\land(F=1)}}_{\Tr}+2\varepsilon^\ast.
    \end{split}
    \end{align}
    To finish the proof, we study the state $\rho_{A'TEY\land(F=1)}$. The cloning attack and the first decryption step takes $\rho_{A'T\bar{M}VSY}$ to $\rho_{A'TEY\land(F=1)}$ via a trace non-increasing channel. Therefore, if we have $\rho_{A'T\bar{M}VS|(Y=0)}=\rho_{A'T\bar{M}VS|(Y=1)}$, then $\rho_{A'TE|(Y=0)\land(F=1)}=\rho_{A'TE|(Y=1)\land(F=1)}$. To that end,
    \begin{align}
        \rho_{A'T\bar{M}VSY}&=\frac{1}{2}\sum_m\expec_{a,t,t',r,h}[at]\otimes[m+e(t',r)+h]\otimes\ketbra{a_{t,t'}}\otimes\parens*{\delta_{m,m_0}\rho_S\otimes[0]+p_m\rho^m_S\otimes[1]}\nonumber\\
        &=\frac{1}{2}\expec_{a,t,t'}[at]\otimes\mu_{\bar{M}}\otimes\ketbra{a_{t,t'}}\otimes\parens*{\rho_S\otimes[0]+\rho_S\otimes[1]}\\
        &=\rho_{A'T\bar{M}V}\otimes\rho_S\otimes\mu_Y\nonumber,
    \end{align}
    so $\rho_{A'T\bar{M}VS|(Y=0)}=\rho_{A'T\bar{M}VS|(Y=1)}$, giving the result.
\end{proof}

\section{Uncloneable Bit Commitment}\label{sec:commitment}

In this section, we discuss our second application, introduced in \cref{sec:intro-commitment}. 
In \cref{sec:commit-defs}, we define uncloneable commitments and provide a construction, given as \cref{prot:urbc}. Finally, in \cref{sec:commit-security}, we prove security of our construction.

\subsection{Motivation and definitions}\label{sec:commit-defs}

We want to extend bit commitment protocols to make them uncloneable --- that is that only the intended recipient can successfully reveal a commitment. First, we recall a usual definition of bit commitment, as in \cite{KWW12}. The form of commitment we use allows for strings, not just single bits, to be committed. Also, it supposes that, in the honest case, a uniformly random string is chosen to be committed; this however is not a restriction on the general case.

\begin{definition}
	A \emph{$(\ell,\varepsilon_1,\varepsilon_2,\varepsilon_3)$-randomised bit string commitment (RBC) scheme} is a pair of interactive protocols between two parties Alice and Bob: a protocol $\ttt{commit}$ that creates a state $\rho_{YAB}$, and a protocol $\ttt{reveal}$ that creates a state $\rho_{YA'\hat{Y}FB'}$. Here $Y=\Z_2^\ell$ is a classical register holding the committed string; $\hat{Y}=\Z_2^\ell$ is a classical register holding the revealed string; $F=\Z_2$ is a classical register that indicates whether Bob accepts (1) or rejects (0) the reveal; and $A,A'$ and $B,B'$ are additional quantum registers that Alice and Bob hold, respectively. The scheme additionally satisfies
	\begin{description}
		\item[$\varepsilon_1$-correctness] If Alice and Bob are honest, then $\norm{\rho_{Y\hat{Y}F}-\sigma_{YYF}}_{\Tr}\leq\varepsilon_1$, for $\sigma_{YF}=\mu_Y\otimes[1]$.
		
		\item[$\varepsilon_2$-hiding] If Alice is honest, then after $\ttt{commit}$, $\norm{\rho_{YB}-\mu_{Y}\otimes\rho_B}_{\Tr}\leq\varepsilon_2$.
		
		\item[$\varepsilon_3$-binding] If Bob is honest, there exists a state $\sigma_{YAB}$ such that $\norm{\rho_{YAB}-\sigma_{YAB}}_{\Tr}\leq\varepsilon_3$, and if $\ttt{reveal}$ is run to get $\sigma_{YA'\hat{Y}FB'}$, $\Pr\squ{Y\neq\hat{Y}\land F=1}_\sigma\leq\varepsilon_3$.
	\end{description}
\end{definition}

Bit commitment is not possible with no additional assumptions \cite{BS16}, so we need a model with, \emph{e.g.}, computational or storage assumptions in order for this definition to not be vacuous. Notwithstanding, we can extend the definition to handle uncloneability as well. We do so by adding an eavesdropper Eve, from whom Alice wishes to hide her commitment. In order to check for cloning, the protocol will have an additional $\ttt{check}$ step which is used to verify whether it is in fact Bob who received the commitment. The separation of the check step also allows us to consider various models: Eve can be allowed to freely communicate with Bob prior to that step, but not afterwards, as Bob could in that case simply give his register that passed the check to her.

\begin{definition}
	A \emph{$(\ell,\varepsilon_1,\varepsilon_2,\varepsilon_3,\delta)$-uncloneable randomised bit string commitment (URBC) scheme} is a triple of protocols between two parties Alice and Bob, eavesdropped by an eavesdropper Eve: a protocol $\ttt{commit}$ that creates a state $\rho_{YABE}$, a protocol $\ttt{check}$ that creates a state $\rho_{YGA'B'E'}$, and a protocol $\ttt{reveal}$ that creates a state $\rho_{YGA''\hat{Y}FB''E''}$. Here, $Y=\Z_2^\ell$ is a classical register holding the committed string; $\hat{Y}=\Z_2^\ell$ is a classical register holding the revealed string; $G=\Z_2$ is a classical register that indicates whether Alice accepts (1) or rejects (0) the check; $F=\Z_2$ is a classical register that indicates whether Bob accepts (1) or rejects (0) the reveal; and $A,A',A''$, $B,B',B''$, and $E,E',E''$ are additional quantum registers that Alice, Bob, and Eve hold, respectively. The scheme additionally satisfies
	\begin{description}
		\item[$\varepsilon_1$-correctness] If Alice and Bob are honest, and Eve does not act, then $\norm{\rho_{YG\hat{Y}F}-\sigma_{YGYF}}_{\Tr}\leq\varepsilon_1$, where $\sigma_{YGF}=\mu_{Y}\otimes[1]\otimes[1]$.
		
		\item[$\varepsilon_2$-hiding] If Alice is honest, then after $\ttt{commit}$, $\norm*{\rho_{YBE}-\mu_Y\otimes\rho_{BE}}_{\Tr}\leq\varepsilon_2$, and after $\ttt{check}$, $\norm*{\rho_{YB'E'}-\mu_Y\otimes\rho_{B'E'}}_{\Tr}\leq\varepsilon_2$.
		
		\item[$\varepsilon_3$-binding] If Bob is honest, there exists a state $\sigma_{YABE}$ such that $\norm{\rho_{YABE}-\sigma_{YABE}}_{\Tr}\leq\varepsilon_3$ and $\Pr\squ{Y\neq\hat{Y}\land F=1}_\sigma\leq\varepsilon_3$.
		
		\item[$\delta$-uncloneability] If Alice is honest, $\norm*{\rho_{YE''\land(G=1)}-\mu_Y\otimes\rho_{E''\land(G=1)}}_{\Tr}\leq\delta$.
	\end{description}
\end{definition}

From this definition, we see that uncloneability holds for any malicious Bob, even one who colludes with Eve, as long as they do not communicate after the check. Similarly to interactive uncloneable encryption, the commitment can be seen as not having an intended recipient prior to the check step --- in particular, Bob and Eve may have arbitrary communication before then. This illustrates an important aspect of the uncloneability, as only Bob will be able to open despite a lack of an agreement between him an Alice, such as a pre-shared secret key.

\begin{remark}
	Note that the above definitions do not hold as given in the computational setting. However, it is straightforward to adapt them by replacing the supremum in the trace norm $\norm{A}_{\Tr}=\sup_{0\leq P\leq\Id}\Tr(PA)$ with the distinguishing advantage corresponding to a computationally-bounded guessing strategy. This allows adaptation to a wide range of computational settings where different computational assumptions that give rise to commitments can be considered. For simplicity, we use the trace norm definition to prove security of our URBC construction, but the proofs work as well in such computational settings simply because the trace norm upper bounds any seminorm given as a supremum over fewer operators. Nevertheless, in our instantiation, the information-theoretic nature of the uncloneability property may be preserved as this does not depend on the choice of commitment assumption.
\end{remark}

Now, we can define a candidate URBC scheme. We do so by taking an RBC scheme and turning it into an uncloneable one on polynomially shorter bit strings using the \gamename MoE property, implicitly working under the assumptions that are required for the commitment.

Let $c=(\ttt{commit}_0,\ttt{reveal}_0)$ be a $(k+\ell,\varepsilon_1,\varepsilon_2,\varepsilon_3)$-RBC scheme, let $A$ be the set of all subspaces of $V=\Z_2^n$ of dimension $n/2$, let $e:\Z_2^{n/2}\times\Z_2^k\rightarrow\Z_2^\ell$ be a quantum-proof $(\kappa,\varepsilon')$-strong extractor, and let $C\subseteq\Z_2^{n/2}$ be an $(n/2,n/2-s,d)$-linear error-correcting code with syndrome $\syn:\Z_2^{n/2}\rightarrow\Z_2^s$.
\begin{mdframed}
\begin{protocol}[Uncloneable bit string commitment]\label{prot:urbc}\hphantom{}

\begin{description}
	\item[Commit] Let $R=\Z_2^k$, $H=\Z_2^\ell$, and $T=T'=\Z_2^{n/2}$. Alice and Bob commit to $(r,h)\in R\times H$ using $c$. Then, Alice samples $a\in A$, $t\in T$, and $t'\in T'$ uniformly at random, after which she prepares the state $\ket{a_{t,t'}}$ and sends it to Bob. Alice stores $t,t',a$ and Bob stores $\ket{a_{t,t'}}$, and they both store what is needed to reveal the commitment of $(r,h)$.
	
	\item[Check] Alice sends Bob $a$ and he measures in the coset state basis to get measurements $\hat{t},\hat{t}'$ of $t,t'$, then sends $\hat{t}$ to Alice. If $\hat{t}=t$, Alice sets $g=1$, else she sets $g=0$. Alice stores $t'$ and Bob stores $\hat{t}'$, and they both store what is needed to reveal the commitment of $(r,h)$.
	
	\item[Reveal] Bob selects a random subset $j\subseteq\{1,\ldots,n/2\}$ of cardinality $\eta n/2$ and sends it to Alice. She replies with $\syn(t')$ and $t'_j$. Then, they reveal the commitment $c$ to get $(\hat{r},\hat{h})$. If $\mathrm{syn}(\hat{t}')=\mathrm{syn}(t')$, $t'_j=\hat{t}'_j$, and $\ttt{reveal}_0$ accepts ($f_0=1$), Bob sets $f=1$; else he sets $f=0$. Alice's output is $e(t',r)+h$ and Bob's output is $e(\hat{t}',\hat{r})+\hat{h}$.
\end{description}
\end{protocol}
\end{mdframed}

This protocol is illustrated in \cref{fig:urbc}.

\begin{figure}[h!]
	\centering
	\begin{subfigure}{\textwidth}
		\centering
		\begin{tikzpicture}
			\draw (-5.8,-0.5) rectangle (-4.2,0.5) node[pos=0.5] {Alice};
			\draw (4.2,-0.5) rectangle (5.8,0.5) node[pos=0.5] {Bob};
			\draw[-Latex] (-4.2,0) -- (4.2,0) node[below, pos=0.5] {$\ket{a_{t,t'}}$};
			\draw (-0.4,0.7) rectangle (0.4,1.5) node[pos=0.5]{$r,h$};
			\draw[dashed,-Latex] (-4.2,0) -- (-0.4,1.1);
			\draw[dotted] (4.2,0) -- (0.4,1.1);
		\end{tikzpicture}
		\label{fig:urbc-commit}
		\caption{The protocol $\ttt{commit}$}
	\end{subfigure}
	
	\begin{subfigure}{\textwidth}
		\centering
		\begin{tikzpicture}
			\draw (-5.8,-0.5) rectangle (-4.2,0.5) node[pos=0.5] {Alice};
			\draw (4.2,-0.5) rectangle (5.8,0.5) node[pos=0.5] {Bob};
			\draw[double,-Latex] (-4.2,0.2) -- (4.2,0.2) node[above, pos=0.5] {$a$};
			\draw[double,Latex-] (-4.2,-0.2) -- (4.2,-0.2) node[below, pos=0.5] {$\hat{t}$};
		\end{tikzpicture}
		\label{fig:urbc-check}
		\caption{The protocol $\ttt{check}$}
	\end{subfigure}
	
	\begin{subfigure}{\textwidth}
		\centering
		\begin{tikzpicture}
			\draw (-5.8,-0.5) rectangle (-4.2,0.5) node[pos=0.5] {Alice};
			\draw (4.2,-0.5) rectangle (5.8,0.5) node[pos=0.5] {Bob};
			\draw (-0.4,0.9) rectangle (0.4,1.7) node[pos=0.5]{$r,h$};
			\draw[dotted] (-4.2,0.2) -- (-0.4,1.3);
			\draw[Latex-,dashed] (4.2,0.2) -- (0.4,1.3);
			\draw[double,Latex-] (-4.2,0.2) -- (4.2,0.2) node[above, pos=0.5] {$j$};
			\draw[double,-Latex] (-4.2,-0.2) -- (4.2,-0.2) node[below, pos=0.5] {$\syn(t'),t'_j$};
			\draw[double,-Latex] (-5,-0.5) -- (-5,-1.5) node[below] {$e(t',r)+h$};
			\draw[double,-Latex] (5,-0.5) -- (5,-1.5) node[below] {$e(\hat{t}',\hat{r})+\hat{h}$};
		\end{tikzpicture}
		\label{fig:urbc-reveal}
		\caption{The protocol $\ttt{reveal}$}
	\end{subfigure}
	\caption{Illustration of the commitment protocol \cref{prot:urbc}. Solid arrows represent transmission of quantum states, double arrows represent transmission of classical information, dashed arrows represent commitment and opening, and dotted lines represent other interactions involved in the commitment without transmission of relevant information.}
	\label{fig:urbc}
\end{figure}
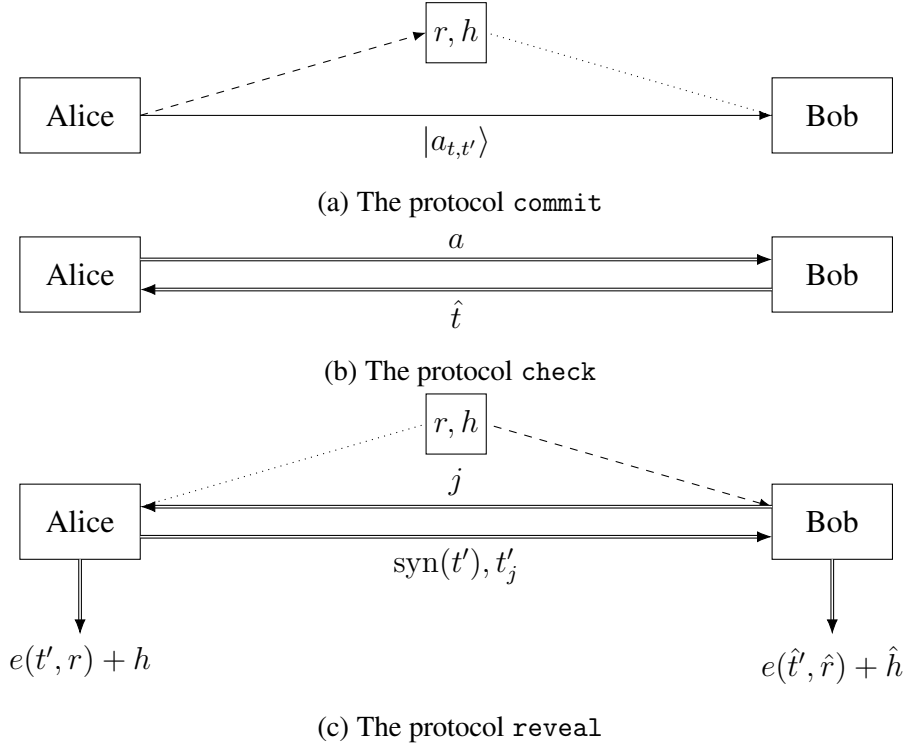

\subsection{Security proofs}\label{sec:commit-security}

\begin{proposition}
    \cref{prot:urbc} is $\varepsilon_1$-correct.
\end{proposition}

\begin{proof}
    We suppose Alice and Bob are honest, and Eve does not act. First, Alice and Bob run $\ttt{commit}_0$ to get $\rho_{RHA_0B_0}$. Then, in the commit and check phases, Alice sends $\ket{a_{t,t'}}$ and $a$ to Bob, and he is able to measure $t,t'$ exactly, so $\hat{t}=t$ and $\hat{t}'=t'$. Bob sends $\hat{t}$ to Alice, and she sets $g=1$. At that point, the shared state has the form $\rho_{RHA_0B_0T'\hat{T}'G}=\rho_{RHA_0B_0}\otimes\sigma_{T'T'}\otimes[1]$ for $\sigma_{T'}=\mu_{T'}$. Next, in the reveal phase, we have that $\mathrm{syn}(\hat{t}')=\mathrm{syn}(t')$ and $\hat{t}_j'=t_j'$, so Bob's flag $f=f_0$. When Alice and Bob run $\ttt{reveal}_0$, the shared state becomes $\rho_{RHA_0'\hat{R}\hat{H}B_0'F_0FT'\hat{T'}G}=\rho_{RHA_0'\hat{R}\hat{H}B_0'F_0F_0}\otimes\sigma_{T'T'}\otimes[1]$, where we know by correctness of $c$ that $\norm{\rho_{RH\hat{R}\hat{H}F_0}-\sigma_{RHRHF_0}}_{\Tr}\leq\varepsilon_1$ for $\sigma_{RHF_0}=\mu_{RH}\otimes[1]$. Thus, for $\sigma_{T'RHF_0}=\mu_{T'RH}\otimes[1]$, we see that
    \begin{align}
    \begin{split}
        &\norm{\rho_{T'\hat{T}'RH\hat{R}\hat{H}F_0F}-\sigma_{T'T'RHRHF_0F_0}}_{\Tr}\leq\norm{\sigma_{T'T'}\otimes(\rho_{RH\hat{R}\hat{H}F_0}-\sigma_{RHRHF_0})}_{\Tr}\leq\varepsilon_1.
    \end{split}
    \end{align}
   	We see that $\sigma_{(e(T',R)+H)F}=\sigma_{YF}=\mu_Y\otimes[1]$, as $\sigma_H=\mu_H$ is hashed. Then, as classical computations are quantum channels,
   	\begin{align}
        \norm{\rho_{(e(T',R)+H)G(e(\hat{T}',\hat{R})+\hat{H})F}-\sigma_{YGYF}}_{\Tr}\leq\norm{\rho_{T'\hat{T}'RH\hat{R}\hat{H}F_0F}-\sigma_{T'T'RHRHF_0F_0}}_{\Tr}\leq\varepsilon_1.
    \end{align}
\end{proof}

\begin{proposition}
    \cref{prot:urbc} is $\varepsilon_2$-hiding.
\end{proposition}

\begin{proof}
    As Alice is honest, the commitment $c$ is hiding in the sense that $\norm{\rho_{RHB_0}-\mu_{RH}\otimes\rho_{B_0}}_{\Tr}\leq\varepsilon_2$. Consider the state $\sigma_{RHATT'VB_0}=\mu_{RH}\otimes\rho_{ATT'VB_0}$. As $H$ is uniformly random, for each $t'\in T'$ and $r\in R$, $e(t',r)+H$ is uniformly random. Hence,
    \begin{align}
    \begin{split}
        \sigma_{(e(T',R)+H)AVB_0}&=\expec_{a,t,t',r,h}[(e(t',r)+h)a]\otimes\ketbra{a_{t,t'}}\otimes\rho_{B_0}=\mu_Y\otimes\rho_{AVB_0}.
    \end{split}
    \end{align}
    As Bob and Eve's registers after $\ttt{commit}$ and $\ttt{check}$ are given by quantum channels acting on $AVB_0$, we get, noting that $\rho_{RHB_0ATT'V}=\rho_{RHB_0}\otimes\rho_{ATT'V}$
    \begin{align}
    \begin{split}
    	\norm{\rho_{YBE}-\mu_Y\otimes\rho_{BE}}_{\Tr}&\leq\norm{\rho_{(e(T',R)+H)AVB_0}-\mu_Y\otimes\rho_{AVB_0}}_{\Tr}\\
    	&=\norm{\rho_{(e(T',R)+H)AVB_0}-\sigma_{(e(T',R)+H)AVB_0}}_{\Tr}\\
    	&\leq\norm{\rho_{RHATT'VB_0}-\sigma_{RHATT'VB_0}}_{\Tr}\\
    	&=\norm{\rho_{ATT'V}\otimes(\rho_{RHB_0}-\mu_{RH}\otimes\rho_{B_0})}_{\Tr}\leq\varepsilon_2.
    \end{split}
    \end{align}
    In the same way $\norm{\rho_{YB'E'}-\mu_Y\otimes\rho_{B'E'}}_{\Tr}\leq\varepsilon_2$.
\end{proof}

\begin{proposition}
    \cref{prot:urbc} is $\varepsilon_3+\parens*{1-\frac{2d}{n}}^{\eta\frac{n}{2}}$-binding.
\end{proposition}

\begin{proof}
	Since $c$ is $\varepsilon_3$-binding, we consider the state $\sigma_{RHA_0B_0}$ such that $\norm{\sigma_{RHA_0B_0}-\rho_{RHA_0B_0}}_{\Tr}\leq\varepsilon_3$. As all the actions undertaken are quantum channels, we know that at the end of the commit phase, $\norm{\sigma_{YABE}-\rho_{YABE}}_{\Tr}\leq\varepsilon_3\leq\varepsilon_3+\parens*{1-\frac{2d}{n}}^{\eta\frac{n}{2}}$. Now, we continue the argument, implicitly assuming that the state is $\sigma$. In the reveal phase, Bob sets $f=1$ if and only if $\mathrm{syn}(t')=\mathrm{syn}(\hat{t}')$, $t'_j=\hat{t}'_j$, and $f_0=1$. Thus,
    \begin{align}
    \begin{split}
        \Pr\squ*{Y\neq\hat{Y}\land F=1}&=\Pr\squ*{e(T',R)+H\neq e(\hat{T}',\hat{R})+\hat{H}\land\mathrm{syn}(T')=\mathrm{syn}(\hat{T}')\land T'_J=\hat{T}'_J\land F_0=1}\\
        &\quad\leq\Pr\squ*{(T'\neq\hat{T}'\lor RH\neq\hat{R}\hat{H})\land\mathrm{syn}(T')=\mathrm{syn}(\hat{T}')\land T'_J=\hat{T}'_J\land F_0=1}\\
        &\quad\leq\Pr\squ*{RH\neq\hat{R}\hat{H}\land F_0=1}+\Pr\squ*{T'\neq \hat{T}'\land\mathrm{syn}(T')=\mathrm{syn}(\hat{T}')\land T'_J=\hat{T}'_J}.
    \end{split}
    \end{align}
    First, as $c$ is binding, $\Pr\squ*{RH\neq\hat{R}\hat{H}\land F_0=1}\leq\varepsilon_3$. Next, suppose that $\mathrm{syn}(t')=\mathrm{syn}(\hat{t}')$ but $t'\neq \hat{t}'$. Then as the code $C$ has distance $d$, the Hamming distance $d(t',\hat{t}')\geq d$. But, as $j$ is a subset of $\eta\frac{n}{2}$ indices chosen uniformly at random, the probability that $t'_j=\hat{t}'_j$ is no more than $\frac{\binom{n/2-d}{\eta n/2}}{\binom{n/2}{\eta n/2}}$. Simplifying,
    \begin{align}
    \begin{split}
        \Pr\big[T'\neq \hat{T}'\land\mathrm{syn}(T')=&\mathrm{syn}(\hat{T}')\land T'_J=\hat{T}'_J\big]\leq\Pr\squ*{T'\neq \hat{T}'\land d(T',\hat{T}')\geq d\land T'_J=\hat{T}'_J}\\
        &\leq\frac{\binom{n/2-d}{\eta n/2}}{\binom{n/2}{\eta n/2}}=\frac{(n/2-d)\cdots(n/2-d-\eta n/2+1)}{(n/2)\cdots (n/2-\eta n/2+1)}\\
        &=\parens*{1-\frac{d}{n/2}}\parens*{1-\frac{d}{n/2-1}}\cdots\parens*{1-\frac{d}{n/2-\eta n/2+1}}\\
        &\leq\parens*{1-\frac{2d}{n}}^{\eta\frac{n}{2}},
    \end{split}
    \end{align}
    which gives the result.
\end{proof}

\begin{theorem}
    Suppose $\kappa\leq\tfrac{-\lg\cos\frac{\pi}{8}}{2}n-\tfrac{1}{4\ln 2}-s-\eta\tfrac{n}{2}$. Then, \cref{prot:urbc} is $\max\{\varepsilon',e^{1/4}(\cos\tfrac{\pi}{8})^{n/2}\}$ -uncloneable.
\end{theorem}

\begin{proof}
    Due to the \gamename MoE property, we must have $H_{\min}(T|AB;T'|A'TE)\geq(-\lg\cos\tfrac{\pi}{8})n-\tfrac{1}{2\ln2}$ when Bob guesses $t$ during the check phase. This implies that, for any measurement $M$ Bob might have made to get $\hat{t}$, either
    $H_{\min}(T|M(AB))_{\rho}\geq\frac{-\lg\cos\frac{\pi}{8}}{2}n-\frac{1}{4\ln 2}$ or $H_{\min}(T'|A'TE)_{\rho_{|(M(AB)=T)}}\geq \frac{-\lg\cos\frac{\pi}{8}}{2}n-\frac{1}{4\ln 2}$. In the former case, the probability that $\hat{t}=t$, and hence that $g=1$, is at most $e^{1/4}(\cos\tfrac{\pi}{8})^{n/2}$. In the latter case, the additional information that Eve gets about $t'$ during the reveal phase is $\mathrm{syn}(t')$ and $t'_j$, so knowing that her final register $E''=A'TE\mathrm{syn}(T')T'_JJ$,
    \begin{align}
    \begin{split}
        H_{\min}(T'|E'')_{\rho_{|(M(AB)=T)}}&\geq H_{\min}(T'|A'TE)_{\rho_{|(M(AB)=T)}}-\lg|\mathrm{syn}(T')|-|J|\\
        &\geq \tfrac{-\lg\cos\frac{\pi}{8}}{2}n-\tfrac{1}{4\ln 2}-s-\eta\tfrac{n}{2}.
    \end{split}
    \end{align}
    Then, by hypothesis on the extractor, $\norm*{\rho_{YE''|(M(AB)=T)}-\mu_Y\otimes\rho_{E''|(M(AB)=T)}}_{\Tr}\leq\varepsilon'$. Thus, combining the two cases and noting that the events $M(AB)=T$ and $G=1$ are equivalent,
    \begin{align}
    \begin{split}
    	&\norm*{\rho_{YE''\land(G=1)}-\mu_Y\otimes\rho_{E''\land(G=1)}}_{\Tr}\\
    	&\qquad=\Pr\squ*{M(AB)=T}\norm*{\rho_{YE''|(M(AB)=T)}-\mu_Y\otimes\rho_{E''|(M(AB)=T)}}_{\Tr}\\
    	&\qquad\leq\max\{\varepsilon',e^{1/4}(\cos\tfrac{\pi}{8})^{n/2}\}.
    \end{split}
    \end{align}
\end{proof} 

\section{Receiver-Independent Quantum Key Distribution}\label{sec:qkd}

In this section, we discuss our final application, introduced in \cref{sec:intro-qkd}. In \cref{sec:robust-leaky}, we prove a version of the leaky MoE property that is robust against errors, given as \cref{thm:robust-stronger-monogamy}, and discuss its expression as an entropic uncertainty relation, given as \cref{cor:robust-entropic-moe}. In \cref{sec:qkd-construction}, we present receiver-independent QKD and provide a construction, given as \cref{prot:qkd}. Finally, in \cref{sec:qkd-security}, we recall the QKD security definitions and prove security for our construction.

\subsection{Robust \gamename MoE property}\label{sec:robust-leaky}

We first need a robust version of the \gamename MoE property, analogous to the game with imperfect guessing in \cite{TFKW13}. To do so, we fix $U,U'\subseteq\Z_2^{n/2}$ to be neighbourhoods of $0$, and modify the \gamename MoE game winning condition by saying that Alice accepts if Bob's answer is in $t+U$ and Charlie's is in $t'+U'$. To warrant the name ``leaky'', we suppose that Charlie gets Bob's potentially erroneous guess of $t$ --- but never the actual value of $t$ chosen by Alice --- before making his guess. In the case of $U=U'=\{0\}$, this reduces to the original \gamename MoE game. We formalise this.

\begin{definition}
	 Let $A$ to be a set of subspaces of $\Z_2^n$ of dimension $n/2$, and $U,U'\subseteq\Z_2^{n/2}$ be neighbourhoods of $0$. A \emph{strategy} $\ttt{S}$ for the $(n,A,U,U')$-robust \gamename monogamy-of-entanglement game is simply a strategy for the $(n,A)$-\gamename MoE game. The \emph{winning probability} of $\ttt{S}$ is
	 \begin{align}
	 	\mfk{w}_{n,A,U,U'}(\ttt{S})=\expec_{a\in A}\expec_{t,t'\in\Z_2^{n/2}}\sum_{u\in U,u'\in U'}\Tr\squ*{(B^a_{t+u}\otimes C^{a,t+u}_{t'+u'})\Phi(\ketbra{a_{t,t'}})}.
	 \end{align}
	 The \emph{optimal winning probability} of $\ttt{G}$ is $\mfk{w}^\ast(n,A,U,U')=\sup_{\ttt{S}}\mfk{w}_{n,A,U,U'}(\ttt{S})$.
\end{definition}

We show an upper bound in a context relevant to QKD, where the errors correspond to independent bit flip errors. We define some standard objects: the Hamming norm of $x\in\Z_2^n$ is the number of non-zero terms, written $|x|$, and the corresponding metric, the Hamming distance, is written $d(x,y)=|x+y|$; the unit ball in of radius $m$ in $\Z_2^n$ is $B(n,m):=\set*{x\in \Z_2^n}{|x|\leq m}$; and the binary entropy function is $h:[0,1]\rightarrow\R$ defined as $h(x)=-x\lg x-(1-x)\lg(1-x)$. We have the very useful bound on the volume of this ball: if $m\leq n/2$, $|B(n,m)|\leq 2^{nh(m/n)}$.

\begin{theorem}\label{thm:robust-stronger-monogamy}
	Let $A$ be the set of register subspaces of $\Z_2^n$ of dimension $n/2$. Then, for $m,m'\leq n/4$
	\begin{align}
		\mfk{w}^\ast(n,A,B(n/2,m),B(n/2,m'))\leq\sqrt{e}2^{\frac{n}{2}h(\frac{2m}{n})+\frac{n}{4}h(\frac{2m'}{n})}\parens*{\cos\tfrac{\pi}{8}}^n.
	\end{align}
\end{theorem}

Note that this bound is not particularly tight. We try to stick with the tightest possible expression throughout the proof before passing to this simple closed-form expression at the very end.

The proof proceeds similarly to \cref{thm:stronger-monogamy}. First, we need a robust generalisation of \cref{lem:coset-overlap}.

\begin{lemma}\label{lem:robust-coset-overlap}
	Let $a,b\subseteq\Z_2^n$ be subspaces of dimension $n/2$, and $U,U'\subseteq\Z_2^{n/2}$ be neighbourhoods of $0$. Then,
    \begin{align}
        \norm*{\sqrt{P^a}\sqrt{P^b}}\leq\max_{t\in\Z_2^n}\parens[\Big]{\abs*{(a+b+t)\cap U_b}|U||U'|\frac{|a\cap b|}{|a|}}^{1/2},
    \end{align}
    where $P^a=\sum_{t,t'\Z_2^{n/2}}\sum_{u\in U,u'\in U'}\ketbra{a_{t,t'}}\otimes B^a_{t+u}\otimes C^{a,t+u}_{t'+u'}$ and $U_b=\set{x_b}{x\in U}\subseteq\Z_2^n$ for $x_b$ as defined in \cref{sec:coset}.
\end{lemma}

\begin{proof}
	Since $\sum_{v'\in U'}C^{b,s+v}_{s'+v'}\leq\Id$ for any $s,s',v\in\Z_2^{n/2}$, we get the bound
	\begin{align}
	\begin{split}
		P^b&\leq\sum_{s,s'\in\Z_2^{n/2};v\in U}\ketbra{b_{s,s'}}\otimes B^b_{s+v}\otimes\Id=\sum_{s\in\Z_2^{n/2};v\in U}\Pi_{b+s_b}\otimes B^b_{s+v}\otimes\Id\\
		&=\sum_{s\in\Z_2^{n/2}}\Pi_{\bigcup_{v\in U}(b+(s+v)_b)}\otimes B^b_{s}\otimes\Id.
	\end{split}
	\end{align}
	Since the right hand side is a projector, we have by monotonicity of the square root that it is also a bound on $\sqrt{P^b}$. We also bound
	\begin{align}
		P^a\leq\sum_{t,t',u,u'}\ketbra{a_{t,t'}}\otimes\Id\otimes C^{a,t+u}_{t'+u'}=\sum_{t,t',u,u'}\ketbra{a_{t+u,t'+u'}}\otimes\Id\otimes C^{a,t}_{t'}.
	\end{align}
	Using these,
	\begin{align}
	\begin{split}
		&\norm*{\sqrt{P^a}\sqrt{P^b}}=\norm{\sqrt{P^b}P^a\sqrt{P^b}}^{1/2}\\
		&\leq\norm[\Big]{\sum_{\substack{t,t',s\in\Z_2^{n/2}\\u\in U,u'\in U'}}\Pi_{\bigcup_{v\in U}(b+(s+v)_b)}\ketbra{a_{t+u,t'+u'}}\Pi_{\bigcup_{v\in U}(b+(s+v)_b)}\otimes B^b_{s}\otimes C^{a,t}_{t'}}^{1/2}\\
		&\leq\max_{s\in\Z_2^{n/2}}\norm[\Big]{\sum_{t,t',u,u'}\Pi_{\bigcup_{v\in U}(b+(s+v)_b)}\ketbra{a_{t+u,t'+u'}}\Pi_{\bigcup_{v\in U}(b+(s+v)_b)}\otimes C^{a,t}_{t'}}^{1/2}.
	\end{split}
	\end{align}
	Next, using the triangle inequality,
	\begin{align}
	\begin{split}
		&\norm*{\sqrt{P^a}\sqrt{P^b}}\leq\max_{s}\parens[\Big]{\sum_u\norm[\Big]{\sum_{t,t',u'}\Pi_{\bigcup_v(b+(s+v)_b)}\ketbra{a_{t+u,t'+u'}}\Pi_{\bigcup_v(b+(s+v)_b)}\otimes C^{a,t}_{t'}}}^{1/2}.
	\end{split}
	\end{align}
	Now, as the terms $\sum_{u'}\Pi_{\bigcup_v(b+(s+v)_b)}\ketbra{a_{t+u,t'+u'}}\Pi_{\bigcup_v(b+(s+v)_b)}\otimes C^{a,t}_{t'}$ of the sum are Hermitian operators with orthogonal supports, we can bound
	\begin{align}
	\begin{split}
		\norm*{\sqrt{P^a}\sqrt{P^b}}&\leq\max_{s}\parens[\Big]{\sum_u\max_{t,t'}\norm[\Big]{\sum_{u'}\Pi_{\bigcup_v(b+(s+v)_b)}\ketbra{a_{t+u,t'+u'}}\Pi_{\bigcup_v(b+(s+v)_b)}\otimes C^{a,t}_{t'}}}^{1/2}\\
		&\leq\max_{s}\parens[\Big]{|U|\max_{t,t'}\norm[\Big]{\sum_{u'}\Pi_{\bigcup_v(b+(s+v)_b)}\ketbra{a_{t,t'+u'}}\Pi_{\bigcup_v(b+(s+v)_b)}}}^{1/2}.
	\end{split}
	\end{align}
	For each of these terms,
	\begin{align}
	\begin{split}
		&\norm[\Big]{\sum_{u'}\Pi_{\bigcup_v(b+(s+v)_b)}\ketbra{a_{t,t'+u'}}\Pi_{\bigcup_v(b+(s+v)_b)}}\\
		&\leq\sum_{u'}\norm[\Big]{\Pi_{\bigcup_v(b+(s+v)_b)}\ketbra{a_{t,t'+u'}}\Pi_{\bigcup_v(b+(s+v)_b)}}\\
		&=\sum_{u'}\braket{a_{t,t'+u'}}{\Pi_{\bigcup_v(b+(s+v)_b)}}{a_{t,t'+u'}}=|U'|\frac{\abs*{(a+t_a)\cap\bigcup_v(b+(s+v)_b)}}{|a|}.
	\end{split}
	\end{align}
	The cardinality of the intersection may be written as
	\begin{align}
	\begin{split}
		\abs[\Big]{(a+t_a)\cap\bigcup_v(b+(s+v)_b)}&=|a\cap b|\abs*{\set{v\in U}{(a+t_a)\cap(b+(s+v)_b)\neq\varnothing}}\\
		&=|a\cap b|\abs*{\set{v\in U}{t_a+s_b+v_b\in a+b}}\\
		&=|a\cap b|\abs*{(a+b+t_a+s_b)\cap U_b}.
	\end{split}
	\end{align}
	This gives the wanted bound
	\begin{align}
	\begin{split}
		\norm*{\sqrt{P^a}\sqrt{P^b}}&\leq\max_{s}\parens[\Big]{|U|\max_{t,t'}|U'||a\cap b|\abs*{(a+b+t_a+s_b)\cap U_b}\frac{|a\cap b|}{|a|}}^{1/2}\\
		&\leq\max_{t\in\Z_2^n}\parens[\Big]{|U||U'|\abs*{(a+b+t)\cap U_b}\frac{|a\cap b|}{|a|}}^{1/2}.
	\end{split}
	\end{align}
\end{proof}

Now, we proceed to the proof of the theorem.

\begin{proof}[Proof of \cref{thm:robust-stronger-monogamy}]
	Write $U=B(n/2,m)$ and $U'=B(n/2,m')$. First, we bound the winning probability by an operator norm
	\begin{align}
		\mfk{w}_{n,A,U,U'}(\ttt{S})\leq\norm[\Big]{\expec_{a\in A}P^a},
	\end{align}
	so that we can apply \cref{lem:sum-bound} using the same permutations $\pi_s:S\rightarrow S$ as in \cref{thm:stronger-monogamy}, giving
	\begin{align}
		\mfk{w}_{n,A,U,U'}(\ttt{S})\leq\expec_{s\in S}\max_{\gamma\in S}\norm*{\sqrt{P^{\spn\gamma}}\sqrt{P^{\spn\pi_s(\gamma)}}}.
	\end{align}
	We use \cref{lem:robust-coset-overlap} to write the overlap $\norm*{\sqrt{P^a}\sqrt{P^b}}$ in terms $\dim(a\cap b)$. Suppose $a=\spn\gamma$ and $b=\spn\eta$. Then $U_b=\set*{u\in\Z_2^n}{u_\eta=0,|u|\leq m}$. Thus, as $a+b=\spn(\eta\cup\gamma)$, for any $t\in\Z_2^n$,
	\begin{align}
		(a+b+t)\cap U_b=\set*{u\in\Z_2^n}{u_{\eta^c\cap\gamma^c}=t_{\eta^c\cap\gamma^c},u_\eta=0,|u|\leq m}.
	\end{align}
	To maximise the cardinality of this set, we take $t_{\eta^c\cap\gamma^c}=0$, so
	\begin{align}
	\begin{split}
		|(a+b+t)\cap U_b|&=\abs*{\set*{u\in\Z_2^n}{u_{\eta\cup\gamma^c}=0,|u|\leq m}}\\
		&=|B(|\eta^c\cap\gamma|,m)|=|B(n/2-\dim(a\cap b),m)|.
	\end{split}
	\end{align}
	This gives $\norm*{\sqrt{P^a}\sqrt{P^b}}\leq\sqrt{|B(n/2,m)||B(n/2,m')||B(n/2-\dim(a\cap b),m)|}2^{\dim(a\cap b)/2-n/4}$. Putting this into the bound on the winning probability,
	\begin{align}
		\mfk{w}_{n,A,U,U'}(\ttt{S})\leq\frac{1}{\binom{n}{n/2}}\sum_{k=0}^{n/2}\binom{n/2}{k}^2\sqrt{|B(n/2,m)||B(n/2,m')||B(k,m)|2^{-k}}.
	\end{align}
	We can bound $B(k,m)\leq B(n/2,m)$ and therefore
	\begin{align}
	\begin{split}
		\mfk{w}_{n,A,U,U'}(\ttt{S})&\leq\frac{|B(n/2,m)|\sqrt{|B(n/2,m')|}}{\binom{n}{n/2}}\sum_{k=0}^{n/2}\binom{n/2}{k}^22^{-k/2}\\
		&\leq\sqrt{e}|B(n/2,m)|\sqrt{|B(n/2,m')|}(\cos\tfrac{\pi}{8})^n.
	\end{split}
	\end{align}
	Using the bound on the volume of a ball $B(n/2,m)\leq 2^{\frac{n}{2}h(\frac{2m}{n})}$ gives the result.
\end{proof}

It will prove useful to express the winning probability of this game as a sequential min-entropy as well.

\begin{corollary}
	Fix a strategy for the $(n,A,U,U')$-robust leaky monogamy game with $U=B(n/2,m)$ and $U'=0$. Let the state
	\begin{align}
		\sigma_{AA'TT'BC}=\expec_{a,t,t',u,u'}[aatt']\otimes\Phi(\ketbra{a_{t+u,t'+u'}}).
	\end{align}
	Then, the sequential min-entropy
	\begin{align}
		H_{\min}(T|AB;T'|A'TC)_{\sigma}\geq\parens*{-\lg\cos\tfrac{\pi}{8}-\tfrac{1}{(2\ln 2)n}}n.
	\end{align}
\end{corollary}

Note that we pack the approximate guessing into the state, so we can derive a result on the sequential min-entropy of that state.

\begin{proof}
	 First, the winning probability may be rewritten as
	\begin{align}
	\begin{split}
		\mfk{w}_{n,A,U,U'}(\ttt{S})&=\expec_{a,t,t'}\Tr\squ[\Big]{(B^a_t\otimes C^{a,t}_{t'})\Phi\parens[\Big]{\sum_{u,u'}\ketbra{a_{t+u,t'+u'}}}}\\
		&=|U||U'|\expec_{a,t,t'}\Tr\squ*{(B^a_t\otimes C^{a,t}_{t'})\sigma^{a,t,t'}_{BC}},
	\end{split}
	\end{align}
	Thus, using the bound $\mfk{w}_{n,A,U,U'}(\ttt{S})\leq\sqrt{e}|U|\sqrt{|U'|}\parens*{\cos\tfrac{\pi}{8}}^n$ of \cref{thm:robust-stronger-monogamy} with $|U'|=1$, $\expec_{a,t,t'}\Tr\squ*{(B^a_t\otimes C^{a,t}_{t'})\sigma^{a,t,t'}_{BC}}\leq\sqrt{e}\parens*{\cos\tfrac{\pi}{8}}^n$. Since this takes a similar form to the winning probability of the original leaky MoE game, we can apply the definition of sequential min-entropy to get the wanted result
	\begin{align}
		H_{\min}(T|AB;T'|A'TC)_{\sigma}\geq\parens*{-\lg\cos\tfrac{\pi}{8}-\tfrac{1}{(2\ln 2)n}}n.
	\end{align}
\end{proof}

Finally, we get an entropic uncertainty relation.

\begin{corollary}[Robust leaky MoE entropic uncertainty relation]\label{cor:robust-entropic-moe}
	For any measurement made by Bob $M:T\rightarrow\mc{P}(AB)$ in the robust leaky MoE game, we have
	\begin{align}
		H_{\min}(T|M(AB))_{\sigma}+H_{\min}(T'|A'TC)_{\sigma_{|(M(AB)=T)}}\geq\parens*{-\lg\cos\tfrac{\pi}{8}-\tfrac{1}{(2\ln2)n}}n.
	\end{align}
\end{corollary}

\subsection{Motivation and construction}\label{sec:qkd-construction}

We consider QKD in a model where neither Bob's classical nor his quantum devices are trusted, and may even have been provided by the eavesdropper, but his communication is trusted. This is a stronger model than the one-sided device independent model considered in \cite{TFKW13}. To illustrate this, we give first an attack against that scheme in this model, which allows the eavesdropper to gain the secret key.

First, we recall the one-sided device-independent QKD protocol given as Figure 1 of \cite{TFKW13}, with one small difference: they considered an entanglement-based model whereas we will work directly in the usual and more practical prepare-and-measure model, knowing that security in the former model implies security in the latter.

\begin{mdframed}
\begin{protocol}[one-sided device independent QKD of \cite{TFKW13}]\hphantom{}

\begin{description}
	\item[State preparation] Alice samples $x\in X=\Z_2^n$ and $\theta\in\Theta=\Z_2^n$ uniformly at random and sends the state $\ket{x^\theta}$ to Bob.
	
	\item[Measurement] Bob confirms receipt of the state, then Alice sends $\theta$ to Bob. He measures to get a string $y$.
	
	\item[Parameter estimation] Alice samples a random subset $T\subseteq\{1,\ldots,n\}$ of size $t$ and sends $T,x_T$ to Bob. If the Hamming distance $d(x_T,y_T)>\gamma n$, Bob aborts.
	
	\item[Error correction] Alice sends an error-correction syndrome $\syn(x_{T^c})$ and a random hash function $F\in\mc{F}$ to Bob. Bob corrects $y_{T^c}$ using the syndrome to get $\hat{x}_{T^c}$
	
	\item[Privacy amplification] Alice computes the output $k=F(x_{T^c})$ and Bob computes $\hat{k}=F(\hat{x}_{T^c})$.
\end{description}	
\end{protocol}
\end{mdframed}

In our model, the security of this QKD scheme can be broken, because we cannot trust Bob's classical device to honestly do parameter estimation. Bob would simply control the communication to and from the device, and receive the message $\hat{k}$ or an abort message once the protocol finishes. Consider the following attack involving a malicious device provided by an eavesdropper Eve. When Alice sends the state $\ket{x^\theta}$, Eve intercepts it and holds on to it, and sends Bob's device $\ket{0^n}$. Then, Eve intercepts every message Alice sends and is able to compute Bob's intended output $\hat{k}$, while Bob's device simply outputs a uniformly random string to him. Neither Alice nor Bob have learned that an attack has happened. In this way, Eve succeeds in completely breaking the security of the one-sided device-independent QKD protocol in the receiver-independent model.

To avoid this sort of attack, we need a QKD protocol where only Bob's communication is trusted but none of his devices are. We present the protocol. Let $e:\Z_2^{n/2}\times R\rightarrow\Z_2^\ell$ be a quantum-proof $(\kappa,\varepsilon)$-strong extractor and $C\subseteq\Z_2^{n/2}$ be a $(n/2,n/2-s,d)$-linear error correcting code with syndrome $\syn:\Z_2^{n/2}\rightarrow\Z_2^s$.


\begin{mdframed}
\begin{protocol}[receiver-independent QKD]\label{prot:qkd}\hphantom{}

\begin{description}
    \item[State preparation] Alice chooses $a\in A$, and $t,t'\in\Z_2^{n/2}$ uniformly at random, then sends the state $\ket{a_{t,t'}}$ to Bob.

    \item[Parameter estimation] Alice sends $a$, and Bob replies with a measurement $\hat{t}$ of $t$. If the distance~${d(\hat{t},t)>\gamma\frac{n}{2}}$, Alice aborts the protocol.

    \item[Error correction] Bob makes a measurement $\hat{t}'$ of $t'$, and sends $\syn(\hat{t}')$ to Alice. She uses it to correct\tablefootnote{The correction here is not simply the natural one of the error-correcting code. Rather, Alice sets $\bar{t}'$ to be the string that corrects to the same point in $C$ as $t'$ but has syndrome $\syn(\hat{t}')$.} $t'$ and get $\bar{t}'$

    \item[Information reconciliation] Alice sends $j\subseteq\{1,\ldots,\tfrac{n}{2}\}$ of cardinality $\eta\tfrac{n}{2}$ to Bob, and he replies with $\hat{t}_{j}$. If $\hat{t}'_j\neq\bar{t}'_j$, Alice aborts.

    \item[Privacy amplification] Alice sends uniformly random $r\in R$ to Bob. Alice outputs ${k=e(\bar{t}',r)}$ and Bob outputs $\hat{k}=e(\hat{t}',r)$.
\end{description}
\end{protocol}
\end{mdframed}

We note that, unlike usual QKD, Alice has full control over whether to abort the protocol. This allows us to consider the case where the checks that Bob makes are untrusted.

Since Bob's classical computations are untrusted, the idea of correctness must also be altered. Neither Alice nor Bob can in general check that Bob's final key matches Alice's, since Bob's device can always, once all the checks have been passed, output a uniformly random string to Bob. As such, all Alice can assure herself of is that Bob's device has all the necessary information allowing it to compute the key. So, we only require correctness to hold for the device's computed key, though Bob may not actually receive it.

\subsection{QKD security}\label{sec:qkd-security}

First, following \cite{MR22arxiv,Ren05}, we give the security definiton of QKD.

\begin{definition}
    A \emph{receiver-independent QKD protocol} is an interaction between Alice, who is trusted, and Bob, who has trusted communication but untrusted quantum and classical devices, and which is eavesdropped by an eavesdropper Eve. The interaction produces the state $\rho_{FK\hat{K}E}$ where  $F=\Z_2$ holds a flag set to $1$ if the protocol accepts and $0$ otherwise, $K=\Z_2^\ell$ holds Alice's outputted key, $\hat{K}=\Z_2^\ell$ holds Bob's device's key, and $E$ is Eve's side information. The protocol is
    \begin{itemize}
        \item \emph{$\varepsilon_1$-correct} if $\Pr\squ*{K\neq\hat{K}\land F=1}\leq\varepsilon_1$.

        \item \emph{$\varepsilon_2$-secret} if $\norm{\rho_{KE\land(F=1)}-\mu_K\otimes\rho_{E\land(F=1)}}_{\Tr}\leq\varepsilon_2$.

        \item \emph{$(\Phi,\varepsilon_3)$-complete} if, when Eve acts as the channel $\Phi$ and Bob's device works as intended, $\Pr\squ*{F=0}\leq\varepsilon_3$.
    \end{itemize}
\end{definition}

A subtle but important difference between this and the usual QKD definition is in Bob's key~$\hat{k}$. Here, the key is produced by Bob's device, but as the device is untrusted, Alice cannot be sure that the key is actually given to Bob at the end of the protocol.

We now show that \cref{prot:qkd} satisfies these security properties under some conditions on the parameters.

\begin{proposition}\label{prop:qkd-correct}
	\cref{prot:qkd} is $\parens*{1-\tfrac{2d}{n}}^{\eta\frac{n}{2}}$-correct.
\end{proposition}

Note that, in our protocol, in order for Bob to actually receive the key, Bob's classical device is only required to do one computation honestly: the final privacy amplification step.

\begin{proof}
	First, the event that $F=1$ is equivalent to $d(T,\hat{T})\leq\gamma\frac{n}{2}\land \bar{T}'_J=\hat{T}'_J$. Then,
	\begin{align}
		\Pr[K\neq\hat{K}\land F=1]\leq\Pr\squ{e(\bar{T}',R)\neq e(\hat{T}',R)\land\bar{T}'_J=\hat{T}'_J}\leq\Pr\squ{\bar{T}'\neq\hat{T}'\land\bar{T}'_J=\hat{T}'_J}
	\end{align}
	We claim that the event $\bar{T}'\neq\hat{T}'$ implies the event $d(\bar{T}',\hat{T}')\geq d$. To see this, (writing for $x\in\Z_2^{n/2}$, $\corr(x)\in C$ the correction from the error-correcting code, \emph{i.e.} the nearest point in $C$ to~$x$) first note that if $\bar{t}'\neq\hat{t}'$, then $\corr(t')=\corr(\bar{t}')\neq\corr(\hat{t}')$. Then, as the code distance is $d$, $d(\corr(t'),\corr(\hat{t}'))\geq d$. Since $\bar{t}'+\corr(t')=\hat{t}'+\corr(\hat{t}')$, $d(\bar{t}',\hat{t}')=d(\corr(t'),\corr(\hat{t}'))\geq d$.
	
	Thus, as $j$ is sampled uniformly at random among the substrings of length $\eta\tfrac{n}{2}$,
	\begin{align}
		\Pr\squ{d(\bar{T}',\hat{T}')\geq d\land\bar{T}'_J=\hat{T}'_J}\leq\frac{\binom{n/2-d}{\eta n/2}}{\binom{n/2}{\eta n/2}}\leq\parens*{1-\frac{2d}{n}}^{\eta\frac{n}{2}}.
	\end{align}
\end{proof}

\begin{theorem}
	Suppose $0\leq\delta\leq\gamma,\frac{2d}{n}$. Then, \cref{prot:qkd} is $(\Phi^{\otimes n},\parens{e^{-(\gamma-\delta)^2}}^n+\parens{e^{-(2d/n-\delta)^2}}^n)$-complete, where $\Phi^{\otimes n}:\mc{L}(V)\rightarrow\mc{L}(V)$ is any iid noise channel such that $\braket{0}{\Phi(\ketbra{1})}{0}\leq\delta$, $\braket{1}{\Phi(\ketbra{0})}{1}\leq\delta$, $\braket{+}{\Phi(\ketbra{-})}{+}\leq\delta$, and $\braket{-}{\Phi(\ketbra{+})}{-}\leq\delta$.
\end{theorem}

In particular, note that this gives an exponentially small abort rate if the error $\delta<\gamma,\frac{2d}{n}$. We make use of Hoeffding's inequality in the proof: for independent random variables $\Gamma_1,\ldots,\Gamma_n$ with support in $[0,1]$, their sum $\Gamma=\sum_i\Gamma_i$ has the property that, for any $x\geq 0$,
\begin{align}
	\Pr\squ*{\Gamma\geq\mbb{E}\Gamma+x}\geq\exp\parens*{-\frac{2x^2}{n}}
\end{align}

\begin{proof}
	First, recall that Alice sends states of the form $\ket{a_{t,t'}}=\ket{x^\theta}$, for $x=t_a+t'_{a^\perp}$ and $\theta=\iota(a)$, the indicator vector. Thus, the conditions on $\Phi$ are simply that there is an independent probability at most $\delta$ of a bit flip on any of the bits of the measured strings. Next, since $\hat{t}'=\bar{t}'$ implies that $\hat{t}'_j=\bar{t}'_j$, we have that
	\begin{align}
	\begin{split}
		\Pr[F=0]&=\Pr\squ*{d(\hat{T},T)>\tfrac{n}{2}\gamma\lor \bar{T}'_J=\hat{T}'_J}\\
		&\leq\Pr\squ*{d(\hat{T},T)>\tfrac{n}{2}\gamma}+ \Pr\squ*{\hat{T}'\neq\bar{T}'}.
	\end{split}
	\end{align}
	First, the probability of more than $\gamma\tfrac{n}{2}$ bit flips occurring on $\hat{t}$ is
	\begin{align}
		\Pr\squ*{d(\hat{T},T)>\tfrac{n}{2}\gamma}\leq\sum_{k=n\gamma/2+1}^{n/2}\binom{n/2}{k}\delta^{k}(1-\delta)^{n/2-k}=\Pr\squ*{\Gamma\geq\gamma\frac{n}{2}+1},
	\end{align}
	where the binomial random variable $\Gamma\sim\mathrm{Bin}(n/2,\delta)$. Consider the independent identically distributed Bernoulli random variables $\Gamma_1,\ldots,\Gamma_{n/2}\sim\mathrm{B}(\delta)$. Since we know $\Gamma=\sum_i\Gamma_i$ and $\mbb{E}\Gamma=\delta\tfrac{n}{2}$, Hoeffding's inequality provides
	\begin{align}
		\Pr\squ*{d(\hat{T},T)>\tfrac{n}{2}\gamma}\leq\exp\parens*{-\frac{4\parens*{(\gamma-\delta)n/2+1}^2}{n}}\leq\parens*{\exp(-(\gamma-\delta)^2)}^n.
	\end{align}
	To proceed similarly for the second term, first note that, in the same way as in \cref{prop:qkd-correct}, $\bar{t}'\neq\hat{t}'$ implies $d(\hat{t}',t')\geq d$. Thus, as before $\Pr\squ*{\hat{T}'\neq\bar{T}'}\leq\Pr\squ*{d(\hat{T}',T')\geq d}\leq\parens*{\exp(-(\tfrac{2d}{n}-\delta)^2)}^n$.
\end{proof}

\begin{lemma}\label{lem:swap-meet}
	Let $X=\Z_2^n$ and $A$ be registers, $\rho_{XA}$ be a cq state, and $U=B(n,m)\subseteq\Z_2^n$ be a ball. For $\sigma=\expec_{u\in U}X_X^u\rho_{XA}X_X^u$ where $X$ is the Pauli operator and any POVM $M:X\rightarrow\mc{P}(A)$, we have
	\begin{align}
		\rho_{M(A)A\land(d(M(A),X)\leq m)}=|U|\sigma_{M(A)A\land(M(A)=X)}.
	\end{align}
\end{lemma}

\begin{proof}
	First, writing $\rho_{XA}=\sum_{x\in X}[x]\otimes\rho^x_A$, we see that
	\begin{align}
		\rho_{XM(A)A}=\sum_{x,y\in X}[xy]\otimes\sqrt{M_y}\rho^x_A\sqrt{M_y},
	\end{align}
	and so
	\begin{align}
		\rho_{M(A)A\land(d(M(A),X)\leq m)}=\sum_{\substack{x,y\in X\\d(x,y)\leq m}}[y]\otimes\sqrt{M_y}\rho^x_A\sqrt{M_y}=\sum_{y\in X}[y]\otimes\sqrt{M_y}\sum_{u\in U}\rho^{y+u}_A\sqrt{M_y}.
	\end{align}
	On the other hand, $|U|\sigma_{XA}=\sum_{x\in X,u\in U}[x]\otimes\rho^{x+u}_A$, so
	\begin{align}
		|U|\sigma_{M(A)A\land(M(A)=X)}=\sum_{x\in X,u\in U}[x]\otimes\sqrt{M_x}\rho^{x+u}_A\sqrt{M_x},
	\end{align}
	which completes the proof.
\end{proof}

\begin{lemma}\label{lem:and-tropy}
	Let $X,Y,A$ be registers and $\rho_{XYA}$ be a ccq state. Then, for any $y_0\in Y$,
	\begin{align}
		H_{\min}(X|A)_{\rho_{\land(Y=y_0)}}\geq H_{\min}(X|AY)_{\rho}.
	\end{align}
\end{lemma}

\begin{proof}
	We interpret this in terms of the guessing probability. Writing $\rho_{XYA}=\sum_{x,y}[xy]\otimes\rho^{x,y}_A$, the probability of guessing $X$ given $AY$ is
	\begin{align}
	\begin{split}
		2^{-H_{\min}(X|AY)_{\rho}}&=\sup_{M^y:X\rightarrow\mc{P}(A)\text{ POVMs}}\sum_{x,y}\Tr\squ*{M^y_x\rho^{x,y}_A}\\
		&\geq\sup_{M:X\rightarrow\mc{P}(A)\text{ POVM}}\sum_{x}\Tr\squ*{M_x\rho^{x,y_0}_A}=2^{-H_{\min}(X|A)_{\rho_{\land(Y=y_0)}}},
	\end{split}
	\end{align}
	as $\rho_{XYA\land(Y=y_0)}=\sum_{x}[xy_0]\otimes\rho^{x,y_0}_A$.
\end{proof}

\begin{theorem}
	Suppose that $\kappa\leq\parens*{-\lg\cos\tfrac{\pi}{8}-\frac{2s}{n}-2\eta-\tfrac{1}{(2\ln 2)n}}\frac{n}{2}$. Then, the QKD protocol \cref{prot:qkd} is $\max\!\!\set{2^{\tfrac{n}{2}h(\gamma)}\varepsilon,2^{-\parens{-\lg\cos\tfrac{\pi}{8}-h(\gamma)-\frac{1}{(2\ln 2)n}}\frac{n}{2}}}$-secret.
\end{theorem}

Asymptotically, in order for the QKD protocol to produce a secure key, we require only
\begin{align}
	\parens*{-\lg\cos\tfrac{\pi}{8}-h(\gamma)-\tfrac{1}{(2\ln 2)n}}\tfrac{n}{2}>0,\parens*{-\lg\cos\tfrac{\pi}{8}-\tfrac{2s}{n}-2\eta-\tfrac{1}{(2\ln 2)n}}\tfrac{n}{2}>0,
\end{align}
as we can make $\varepsilon$ arbitrarily small by enlarging the key. These provide the asymptotic noise tolerance. First $\frac{1}{2\ln 2 n}\rightarrow 0$ and we can choose $\eta$ small enough to have $\eta\rightarrow0$ while preserving subexponential correctness (for example $\eta=1/\sqrt{n}$), so we don't need to worry about those terms. Also, the Shannon limit provides the minimum value $s=\frac{n}{2}h(\gamma)$. Therefore, the inequalities reduce to $-\lg\cos\tfrac{\pi}{8}>h(\gamma)$ asymptotically, so approximately $\gamma<0.0153$; thus the asymptotic noise tolerance is $\approx1.5\%$. Note that this is the same tolerance as in \cite{TFKW13}.

\begin{proof}
	At the start of the protocol, Alice prepares the state $\rho_{ATT'V}=\expec_{a,t,t'}[att']\otimes\ketbra{a_{t,t'}}$, where she holds onto $ATT'$ and sends $V$. Eve acts with some channel $\Phi:\mc{L}(V)\rightarrow\mc{L}(BE)$ and sends the register $B$ to Bob. Bob sends $\hat{T}$ to Alice, which Eve may intercept and copy. We work first with the state $\sigma_{ATT'BE}=\expec_{u\in U}X_T^u\rho X_T^u$, and then exchange it for $\rho$ later, using \cref{lem:swap-meet}. At the parameter estimation step, the robust \gamename MoE property implies $H_{\min}(T|AB;T'|A'TE)_\sigma\geq\parens*{-\lg\cos\tfrac{\pi}{8}-\tfrac{1}{(2\ln 2)n}}n$, where $A'$ is a copy of $A$. Let $M:T\rightarrow\mc{P}(AB)$ be the measurement Bob's device uses to get the guess of $T$. Then, by the entropic uncertainty relation, we must have either
	\begin{align}
	\begin{split}
		&H_{\min}(T|M(AB))_\sigma\geq\parens*{-\lg\cos\tfrac{\pi}{8}-\tfrac{1}{(2\ln 2)n}}\frac{n}{2}\qquad\text{or}\\ &H_{\min}(T'|A'TE)_{\sigma_{|(M(AB)=T)}}\geq\parens*{-\lg\cos\tfrac{\pi}{8}-\tfrac{1}{(2\ln 2)n}}\frac{n}{2}.
	\end{split}
	\end{align}
	In the former case, we have
	\begin{align}
		\Tr\parens[\big]{\sigma_{\land\parens*{\hat{T}=T}}}=\Pr\squ[\big]{M(AB)=T}_\sigma\leq 2^{-\parens[\big]{-\lg\cos\tfrac{\pi}{8}-\tfrac{1}{(2\ln 2)n}}\frac{n}{2}},
	\end{align}
	as $\hat{T}=M(AB)$. In the latter case, by the error correction step, Eve holds $E_0=A'\hat{T}\syn(\hat{T}')J\hat{T}'_JE$ and thus, making use of \cref{lem:and-tropy}
	\begin{align}
	\begin{split}
		H_{\min}(T'|A'\hat{T}\syn(\hat{T}')J\hat{T}'_JE)_{\sigma_{|(\hat{T}=T)\land(\hat{T}'_J=\bar{T}'_J)}}&\geq H_{\min}(T'|A'\hat{T}\syn(\hat{T}')J\hat{T}'_J\bar{T}'_JE)_{\sigma_{|(\hat{T}=T)}}\\
		&\geq H_{\min}(T'|A'TE)_{\sigma_{|(M(AB)=T)}}-s-2\eta\tfrac{n}{2}\\
		&\geq\parens*{-\lg\cos\tfrac{\pi}{8}-\tfrac{2s}{n}-2\eta-\tfrac{1}{(2\ln 2)n}}\tfrac{n}{2}.
	\end{split}
	\end{align}
	Next, as Eve has access to the syndrome $\syn(\hat{t}')$, her probability of guessing $t'$ is equal to that of guessing $\bar{t}'$, giving $H_{\min}(\bar{T}'|A\hat{T}\syn(\hat{T}')E)_{\sigma_{|(\hat{T}=T)\land(\hat{T}'_J=\bar{T}'_J)}}\geq\parens*{-\lg\cos\tfrac{\pi}{8}-\frac{2s}{n}-2\eta-\tfrac{1}{(2\ln 2)n}}\frac{n}{2}$. By hypothesis on the strong extractor, we have that
	\begin{align}
		\norm{\sigma_{e(\bar{T}',R)RE_0|(\hat{T}=T)\land(\hat{T}'_J=\bar{T}'_J)}-\mu_{Z}\otimes\mu_{R}\otimes\sigma_{E_0|(\hat{T}=T)\land(\hat{T}'_J=\bar{T}'_J)}}_{\Tr}\leq\varepsilon,
	\end{align}
	where the register $Z=\Z_2^\ell$. Before passing to the information reconciliation step, we combine the two cases. Writing $\varepsilon^\ast=\max\set{\varepsilon,2^{-\parens[\big]{-\lg\cos\tfrac{\pi}{8}-\tfrac{1}{(2\ln 2)n}}\frac{n}{2}}}$, we get
	\begin{align}
	\begin{split}
		&\norm{\sigma_{e(\bar{T}',R)RE_0\land(\hat{T}=T\land\hat{T}'_J=\bar{T}'_J)}-\mu_{Z}\otimes\mu_{R}\otimes\sigma_{E_0\land(\hat{T}=T\land\hat{T}'_J=\bar{T}'_J)}}_{\Tr}\\
		&\qquad=\Tr(\sigma_{\land(\hat{T}=T)})\norm{\sigma_{e(\bar{T}',R)RE_0|(\hat{T}=T)\land(\hat{T}'_J=\bar{T}'_J)}-\mu_{Z}\otimes\mu_{R}\otimes\sigma_{E_0|(\hat{T}=T)\land(\hat{T}'_J=\bar{T}'_J)}}_{\Tr}\leq\varepsilon^\ast.
	\end{split}
	\end{align}
	Now, we can pass to the real state $\rho$. Using \cref{lem:swap-meet} with $X=T$,
	\begin{align}
	\begin{split}
		&\norm{\rho_{e(\bar{T}',R)RE_0\land(d(\hat{T},T)\leq\gamma n/2\land\hat{T}'_J=\bar{T}'_J)}-\mu_{Z}\otimes\mu_{R}\otimes\rho_{E_0\land(d(\hat{T},T)\leq\gamma n/2\land\hat{T}'_J=\bar{T}'_J)}}_{\Tr}\\
		&\qquad=|U|\norm{\sigma_{e(\bar{T}',R)RE_0\land(\hat{T}=T\land\hat{T}'_J=\bar{T}'_J)}-\mu_{Z}\otimes\mu_{R}\otimes\sigma_{E_0\land(\hat{T}=T\land\hat{T}'_J=\bar{T}'_J)}}_{\Tr}\leq2^{\tfrac{n}{2}h(\gamma)}\varepsilon^\ast.
	\end{split}
	\end{align}
	As the event $F=1$ is equivalent to $d(\hat{T},T)\leq\gamma n/2\land\hat{T}'_J=\bar{T}'_J$, this means
	\begin{align}
		\norm{\rho_{e(\bar{T}',R)RE_0\land(F=1)}-\mu_{K}\otimes\rho_{RE_0\land(F=1)}}_{\Tr}\leq 2^{\tfrac{n}{2}h(\gamma)}\varepsilon^\ast.
	\end{align}
	Finally, as Eve's register at the end of the privacy amplification step is $E'=RE_0=RA\hat{T}\syn(\hat{T}')J\hat{T}_JE$, we get the wanted result $\norm{\rho_{KE'\land(F=1)}-\mu_{K}\otimes\rho_{E'\land(F=1)}}_{\Tr}\leq 2^{\tfrac{n}{2}h(\gamma)}\varepsilon^\ast$.
\end{proof} 

\fi

\newpage


\bibliographystyle{bibtex/bst/alphaarxiv.bst}
\bibliography{bibtex/bib/full.bib,bibtex/bib/quantum.bib,bibtex/quantum_new.bib}

\end{document}